\documentclass[backref=page]{article}
\usepackage{ijcai22}

\usepackage{times}
\usepackage{soul}
\usepackage{url}
\usepackage[hidelinks]{hyperref}
\usepackage[utf8]{inputenc}
\usepackage[small]{caption}
\usepackage{graphicx}
\usepackage{amsmath}
\usepackage{amsthm}
\usepackage{booktabs}
\title{Fourier Analysis-based Iterative Combinatorial Auctions\footnotemark[1]}

\author{
Jakob Weissteiner$^{1,4}$\footnotemark[2]
\and
Chris Wendler$^{2,4}$\footnotemark[2]\and
Sven Seuken$^{1,4}$\and
Ben Lubin$^3$\And
Markus P{\"u}schel$^{2,4}$
\affiliations
$^1$University of Zurich\\
$^2$ETH Zurich\\
$^3$Boston University\\
$^4$ ETH AI Center\\
\emails
weissteiner@ifi.uzh.ch,
chris.wendler@inf.ethz.ch,
seuken@ifi.uzh.ch,
blubin@bu.edu,
pueschel@inf.ethz.ch
}

\usepackage{our_commands_arXiv} 

\begin{document}

\maketitle

\begin{abstract}
    Recent advances in Fourier analysis have brought new tools to efficiently represent and learn set functions. In this paper, we bring the power of Fourier analysis to the design of combinatorial auctions (CAs). The key idea is to approximate bidders' value functions using Fourier-sparse set functions, which can be computed using a relatively small number of queries. Since this number is still too large for practical CAs, we propose a new hybrid design: we first use neural networks (NNs) to learn bidders’ values and then apply Fourier analysis to the learned representations. On a technical level, we formulate a Fourier transform-based winner determination problem and derive its mixed integer program formulation. Based on this, we devise an iterative CA that asks Fourier-based queries. We experimentally show that our hybrid ICA achieves higher efficiency than prior auction designs, leads to a fairer distribution of social welfare, and significantly reduces runtime. With this paper, we are the first to leverage Fourier analysis in CA design and lay the foundation for future work in this area. Our code is available on GitHub: \url{https://github.com/marketdesignresearch/FA-based-ICAs}.
\end{abstract}

\renewcommand{\thefootnote}{\fnsymbol{footnote}}
\footnotetext[1]{This paper is the slightly updated version of \cited{weissteiner2022fourier} published at IJCAI'22 including the appendix.}
\footnotetext[2]{These authors contributed equally to this paper.}
\renewcommand{\thefootnote}{\arabic{footnote}}

\section{Introduction}\label{Introduction}

Combinatorial auctions (CAs) are used to allocate multiple heterogeneous items to bidders. CAs are particularly useful  in domains where bidders' preferences exhibit \textit{complementarities}  and \textit{substitutabilities} as they allow bidders to submit bids on \textit{bundles} of items rather than on individual items.

Since the bundle space grows exponentially in the number of items, it is impossible for bidders to report values for all bundles in settings with more than a modest number of items. Thus, parsimonious preference elicitation is key for the practical design of CAs.
For general value functions, \cited{nisan2006communication} have shown that to guarantee full efficiency, exponential communication in the number of items is needed. Thus, practical CAs cannot provide efficiency guarantees in large domains. Instead, recent proposals have focused on \textit{iterative combinatorial auctions (ICAs)}, where the auctioneer interacts with bidders over rounds, eliciting a \textit{limited} amount of information, aiming to find a highly efficient allocation.

ICAs have found widespread application; most recently, for the sale of licenses to build offshore wind farms \cite{ausubel2011auction}. For the sale of spectrum licenses, the combinatorial clock auction (CCA) \cite{ausubel2006clock} has generated more than \$20 billion in total revenue \cite{ausubel2017practical}. Thus, increasing the efficiency by only 1--2\% points translates into monetary gains of millions of dollars.

\subsection{Machine Learning-based Auction Design}\label{Machine Learning-based CAs}
Recently, researchers have used machine learning (ML) to improve the performance of CAs. Early work by \cited{blum2004preference} and \cited{lahaie2004applying} first studied the relationship between learning theory and preference elicitation in CAs. \cited{dutting2019optimal}, \cited{shen2019automated} and \cited{rahme2020permutation} used neural networks (NNs) to learn whole auction mechanisms from data.
\cited{brero2019fast} introduced a Bayesian ICA using probabilistic price updates to achieve faster convergence. 
\cited{shen2020reinforcement} use reinforcement learning for dynamic pricing in sponsored search auctions. Most related to the present paper is the work by \citeauthor{brero2018combinatorial} \shortcite{brero2018combinatorial,brero2021workingpaper}, who developed a value-query-based ML-powered ICA using support vector regressions (SVRs) that achieves even higher efficiency than the CCA. In follow-up work, \cited{weissteiner2020deep} extended their ICA to NNs, further increasing the efficiency. In work subsequent to the first version of this paper, \cited{weissteiner2022monotone} proposed Monotone-Value Neural Networks (MVNNs), which are particularly well suited to learning value functions in combinatorial assignment domains.\footnote{After the publication of the present paper, \cited{weissteiner2023bayesian} integrated a notion of uncertainty \cite{heiss2022nomu} into an ML-powered ICA to balance the explore-exploit dilemma.}

However, especially in large domains, it remains a challenge to find the efficient allocation while keeping the elicitation cost low. Thus, even state-of-the-art approaches suffer from significant efficiency losses and often result in unfair allocations, highlighting the need for better preference elicitation algorithms.

\subsection{Combining Fourier Analysis and CAs}\label{IncorporatingFourierAnalysisInCADesign}
The goal of preference elicitation in CAs is to learn bidders' value functions using a small number of queries. Mathematically, value functions are \emph{set functions}, which are in general exponentially large and notoriously hard to represent or learn. To address this complexity, we leverage Fourier analysis for set functions \cite{bernasconi1996fourier,o2014analysis,puschel2020discrete}. In particular, we consider \textit{Fourier-sparse approximations}, which are represented by few parameters. These parameters are the non-zero Fourier coefficients (FCs) obtained by a base change with the Fourier transform (FT). We use the framework by \cited{puschel2020discrete}, which contains new FTs beyond the classical Walsh-Hadamard transform (WHT) \cite{bernasconi1996fourier}, providing more flexibility. Until recently, methods for learning Fourier-sparse set functions focused on the WHT, and they placed assumptions on bidders' value functions that are too restrictive for CAs~\cite{stobbe2012learning}. However, recently, \cited{amrollahi2019efficiently} proposed a new algorithm that can approximate general set functions by WHT-sparse ones, which is suitable for large CAs and we use it in this work.

\subsection{Our Contribution} 
Our main contribution in this paper is to bring the power of Fourier analysis to CA design (Section~\ref{FourierAnalysisOfValueFunctions}). In particular, we formulate \emph{FT-based winner determination problems (WDPs)} and derive corresponding mixed integer programs (MIPs) for several FTs (Section~\ref{FourierTransformBasedWDPsChapter}). Our MIPs allow for the efficient solution of the FT-based WDP and provide the foundation for using Fourier-sparse approximations in auction design. 

We first experimentally show that the WHT performs best among the FTs in terms of induced level of sparsity (Section~\ref{Fourier Sparsity of LSVM}) and reconstruction error (Section~\ref{LearningPerformance}). As an initial approach, we develop a WHT-based allocation rule (Section~\ref{AFourierTransformbasedAllocationRule}). However, this requires too many queries for direct use in CAs.
To overcome this, we propose a practical hybrid ICA based on NNs \textit{and} Fourier analysis (Section~\ref{TheHybridMechanism}). The key idea is to compute Fourier-sparse approximations of NN-based bidder representations, enabling us to keep the number of queries small. The advantage of the NN-based representations is that they capture key aspects of the bidders' value functions and can be queried arbitrarily often (Section~\ref{NNs Support Discovery Experiments}). 

Our efficiency experiments show that our hybrid ICA achieves higher efficiency than state-of-the-art mechanisms, leads to a significant computational speedup, and yields fairer
allocations (Section~\ref{Experiments}). This shows that leveraging Fourier analysis in CA design is a promising new research direction.

\section{Preliminaries}\label{Preliminaries}

In this section, we present our formal model and review the MLCA mechanism, which our hybrid ICA builds upon.

\subsection{Formal Model for ICAs}
We consider a CA with $n$ bidders and $m$ indivisible items. Let $N=\{1,\ldots,n\}$ and $M=\{1,\ldots,m\}$ denote the set of bidders and items, respectively. We denote with $x\in \X=\{0,1\}^m$ a bundle of items represented as an indicator vector, where $x_{j}=1$ iff item $j \in M$ is contained in $x$. Bidders' true preferences over bundles are represented by their (private) value functions $v_i: \X\to \R_+,\,\, i \in N$, i.e., $v_i(x)$ represents bidder $i$'s true value for bundle $x$. By $a=(a_1,\ldots,a_n) \in \X^n$ we denote an allocation of bundles to bidders, where $a_i$ is the bundle bidder $i$ obtains. We denote the set of \textit{feasible} allocations by $\F=\left\{a \in \X^n:\sum_{i \in N}a_{ij} \le 1, \,\,\forall j \in M\right\}$. The (true) \textit{social welfare} of an allocation $a$ is defined as $V(a)=\sum_{i \in N} v_i(a_i).$ We let $a^* \in \argmax_{a \in {\F}}V(a)$ be a social-welfare maximizing, i.e., \textit{efficient}, allocation. The efficiency of any $a \in \F$ is measured by $V(a)/V(a^*)$. We assume that bidders' have quasilinear utilities $u_i$, i.e, for a payments $p\in \R^n_+$ it holds that $u_i(a,p) = v_i(a_i) - p_i$.

An ICA \textit{mechanism} defines how the bidders interact with the auctioneer and how the final allocation and payments are determined. We denote a bidder's (possibly untruthful) reported value function by $\pvi{}:\X\to\R_+$. In this paper, we consider ICAs that ask bidders iteratively to report their value $\pvi{x}$ for particular bundles $x$ selected by the mechanism (for early work on value queries see \cite{hudson2003using}). A finite set of such reported bundle-value pairs of bidder $i$ is denoted as ${R_i=\left\{\left(x^{(l)},\pvi{x^{(l)}}\right)\right\},\,x^{(l)}\in \X}$. Let $R=(R_1,\ldots,R_n)$ denote the tuple of reported bundle-value pairs obtained from all bidders. We define the \textit{reported social welfare} of an allocation $a$ given $R$ as $\hV{a|R}=\sum_{i \in N:\, \left(a_i,\pvi{a_i}\right)\in R_i}\pvi{a_i},$
where $\left(a_i,\pvi{a_i}\right)\in R_i$ ensures that only values for reported bundles contribute. Finally, the optimal allocation $a^*_{R}\in \F$ given the reports $R$ is defined as 
\begin{align}\label{WDPFiniteReports}
a^*_{R} \in \argmax_{a \in {\F}}\hV{a|R}.
\end{align}

The final allocation $a^*_{R}\in \F$ and payments $p(R)\in \R^n_+$ are computed based on the elicited reports $R$ \emph{only}.

As the auctioneer can only ask each bidder $i$ a limited number of queries $|R_i| \leq \Qmax$, the ICA needs a smart preference elicitation algorithm, with the goal of finding a highly efficient $a^*_{R}$ with a limited number of value queries.

\subsection{A Machine Learning-powered ICA}\label{AnNNPoweredICAMechanism}
We now review the \textit{machine learning-powered combinatorial auction (MLCA)} by \cited{brero2021workingpaper}. Interested readers are referred to \Cref{subsec:appendix_AnNNPoweredICAMechanism}, where we present MLCA in detail.

MLCA starts by asking each bidder value queries for $\Qinit$ randomly sampled initial bundles. Next, MLCA proceeds in rounds until a maximum number of value queries per bidder $\Qmax$ is reached. In each round, for each bidder $i\in N$, it trains an ML algorithm $\mathcal{A}_i$ on the bidder's reports $R_i$. Next, MLCA generates new value queries $\qnew=(\qnew_i)_{i=1}^{n}$ with $\qnew_i\in \X \setminus R_i$ by solving a ML-based WDP $\qnew \in \argmax\limits_{a \in {\F}}\sum\limits_{i \in N} \mathcal{A}_i(a_i)$.  The idea is the following: if $\mathcal{A}_i$ are good surrogate models of the bidders' true value functions then $\qnew$ should be a good proxy of the efficient allocation $a^*$ and thus provide valuable information.

At the end of each round, MLCA receives reports $\Rnew$ from all bidders for the newly generated $\qnew$ and updates $R$. When $\Qmax$ is reached, MLCA computes an allocation $a^*_R$ maximizing the \emph{reported} social welfare (\cref{WDPFiniteReports}) and determines VCG payments $p(R)$ (see \Cref{subsec:appendix_VCGpayments}).

\subsection{Incentives of MLCA and Hybrid ICA}
A key concern in the design of ICAs are bidders' incentives. However, the seminal result by \cited{nisan2006communication} discussed above implies that practical ICAs cannot simply use VCG to achieve strategyproofness. And in fact, no ICA deployed in practice is \textit{strategyproof} -- including the famous SMRA and CCA auctions used to conduct spectrum auctions.  Instead, auction designers have designed mechanisms that, while being manipulable, have ``good incentives in practice'' (see \cite{cramton2013spectrumauctions,Milgrom2007PackageAuctionsAndExchanges}). 

Naturally, the MLCA mechanism is also not strategyproof, and \cited{brero2021workingpaper} provide a simple example of a possible manipulation. The idea behind the example is straightforward: if the ML algorithm does not learn a bidder's preferences perfectly, a sub-optimal allocation may result. Thus, a bidder may (in theory) benefit from misreporting their preferences with the goal of ``correcting'' the ML algorithm, so that, with the misreported preferences, the mechanism actually finds a preferable allocation. 

However, MLCA has two features that mitigate manipulations. First, MLCA explicitly queries each bidder's marginal economy, which implies that the marginal economy term of the final VCG payment is practically independent of bidder $i$'s bid (for experimental support see \cite{brero2021workingpaper}). Second, MLCA enables bidders to ``push'' information to the auction which they deem useful. This mitigates certain manipulations of the main economy term in the VCG payment rule, as it allows bidders to increase the social welfare directly by pushing (useful) truthful information, rather than attempting to manipulate the ML algorithm. \cited{brero2021workingpaper} argued that with these two design features, MLCA exhibits very good incentives in practice. They performed a computational experiment, testing whether an individual bidder (equipped with more information than he would have in a real auction) can benefit from deviating from truthful bidding, while all other bidders are truthful. In their experiments, they could not identify a beneficial manipulation strategy. While this does not rule out that some (potentially more sophisticated) beneficial manipulations do exist, it provides evidence to support the claim that MLCA has good incentives in practice.

With two additional assumptions, one also obtains a theoretical incentive guarantee for MLCA. Assumption 1 requires that, if all bidders bid truthfully, then MLCA finds an efficient allocation (we show in \Cref{subsec:appendix_efficiency_experiments} that in two of our domains, we indeed find the efficient allocation in the majority of cases). Assumption 2 requires that, for all bidders $i$, if all other bidders report truthfully, then the social welfare of bidder $i$'s marginal economy is independent of his value reports. If both assumptions hold, then bidding truthfully is an ex-post Nash equilibrium in MLCA.

Our hybrid ICA (\Cref{HybridICA} in \Cref{TheHybridMechanism}) is built upon MLCA, leaving the general framework in place, and only changing the algorithm that generates new queries each round. Given this design, the incentive properties of MLCA extend to the hybrid ICA. Specifically, our hybrid ICA is also not strategyproof, but it also has the same design features (including push-bids) to mitigate manipulations.

In future work, it would be interesting to evaluate experimentally whether the improved performance of the hybrid ICA translates into better manipulation mitigation compared to MLCA. However, such an analysis is beyond the scope of the present paper, which focuses on the ML algorithm that is integrated into the auction mechanism.

\section{Fourier Analysis of Value Functions}\label{FourierAnalysisOfValueFunctions}
 We now show how to apply Fourier analysis to value functions providing the theoretical foundation of FT-based WDPs.

Classic Fourier analysis decomposes an audio signal or image into an orthogonal set of sinusoids of different frequencies. Similarly, the classical Fourier analysis for \emph{set functions} (i.e., functions mapping subsets of a discrete set to a scalar) decomposes a set function into an orthogonal set of Walsh functions \cite{bernasconi1996fourier}, which are piecewise constant with values $1$ and $-1$ only. Recent work by \cited{puschel2020discrete} extends the Fourier analysis for set functions with several novel forms of set Fourier transforms (FTs). Importantly, because bidders' value functions are set functions, they are amenable to this type of Fourier analysis, and it is this connection that we will leverage in our auction design.

\mypar{Sparsity} 
The motivation behind our approach is that we expect bidders' value functions to be \emph{sparse}, i.e., they can be described with much less data than is contained in the exponentially-sized full value function. While this sparsity may be difficult to uncover when looking at bidders' value reports directly, it may reveal itself in the Fourier domain (where then most FCs are zero). As all FTs are changes of basis, each FT provides us with a new \emph{lens on the bidder's value function}, revealing structure and thus potentially reducing dimensionality.

\mypar{Set function Fourier transform} 
We now provide a formal description of FTs for reported value functions $\pvi{}$. To do so, we represent $\pvi{}$ as a vector $(\pvi{x})_{x \in \X}$. Each known FT is a change of basis and thus can be represented by a certain matrix $F\in \{-1,0,1\}^{2^m\times 2^m}$ with the form:
\begin{equation}\label{FTformula}
\hvic(y) = (F \pvi{})({y}) = \sum_{x \in \X} F_{y, x} \pvi{x}.
\end{equation}
There is exactly one Fourier coefficient per bundle, this follows from the theory presented by \cited{puschel2020discrete}.
The corresponding inverse transform $F^{-1}$ is thus:
\begin{equation}\label{eq:FTINVformula}
\pvi{x} = (F^{-1} \hvic)(x) = \sum_{y \in \X} F^{-1}_{x, y} \hvic(y).
\end{equation}
$\hvic$ is again a set function and we call $\hvic(y)$ the \emph{Fourier coefficient} at frequency $y$. A value function is \emph{Fourier-sparse} if $|\supp(\hvic)| = |\set{y: \hvic(y) \neq 0}| \ll 2^m$. We call $\supp(\hvic)$ the \emph{Fourier support} of $\pvi{}$.

Classically, the WHT is used as $F$ \cite{bernasconi1996fourier,o2014analysis}, but we also consider two recently introduced FTs (FT3, FT4) due to their information-theoretic interpretation given in \cite{puschel2020discrete}:
\begin{align}
\textrm{\textbf{FT3:}}\hspace{0.5cm}&F_{y,x} = (-1)^{|y| - |x|} \mathbb{I}_{\min(x,y) = x},\label{FT3formula}&\\
\textrm{\textbf{FT4:}}\hspace{0.5cm}&F_{y,x} = (-1)^{|\min(x,y)|} \mathbb{I}_{\max(x,y) = 1_m},\label{FT4formula}&\\
\textrm{\textbf{WHT:}}\hspace{0.5cm}&F_{y,x} = \frac{1}{2^m} (-1)^{|\min(x, y)|}.\label{WHTformula}
\end{align}
Here, $\min$ is the elementwise minimum (intersection of sets), $\max$ analogously, $|\cdot|$ is the set size, $1_m$ denotes the $m$-dimensional vector of 1s, and the indicator function $\mathbb{I}_{P}$ is equal to $1$ if the predicate $P$ is true and $0$ otherwise.

\mypar{Notions of Fourier-sparsity} 
In recent years, the notion of Fourier-sparsity has gained considerable attention, leading to highly efficient algorithms to compute FTs \cite{stobbe2012learning,amrollahi2019efficiently,wendler2020learning}. Many classes of set functions are Fourier-sparse (e.g., graph cuts, hypergraph valuations and decision trees \cite{abraham2012combinatorial}) and can thus be learned efficiently. The benefit of considering multiple FTs is that they offer different, non-equivalent notions of sparsity as illustrated by the following example.
\begin{example}
Consider the set of items $M = \set{1, 2, 3}$ and the associated reported value function $\pvi{}$ shown in \Cref{tab:example} (where we use $001$ as a shorthand notation for $(0,0,1)$), together with the corresponding FCs $\hvic$:
\begin{table}[h!]
\centering
\begin{sc}
\scriptsize
	\begin{tabular}{
			c
			S[table-format=1.0]
			S[table-format=1.0]
			S[table-format=1.0]
			S[table-format=1.0]
			S[table-format=1.0]
			S[table-format=1.0]
			S[table-format=1.0]
			S[table-format=1.0]
			}
		\toprule
		& {000} & {100} & {010} & {001}& {110} & {101} & {011} & {111}\\
		\cmidrule(lr){2-2}
		\cmidrule(lr){3-5}
        \cmidrule(lr){6-8}
        \cmidrule(lr){9-9}
		$\pvi{}$   & {0}& {\vphantom{-}1} & {\vphantom{-}1}  & {\vphantom{-}1}  & {3} & {3} & {3} & {\vphantom{-}5}\\
		{FT3}      & {0}& {\vphantom{-}1}& {\vphantom{-}1} & {\vphantom{-}1} & {1} & {1} & {1} & {-1} \\
		{FT4}      & {5}& {-2}& {-2} & {-2} & {0} & {0} & {0} & {\vphantom{-}1} \\
		{WHT}      & {$\nicefrac{17}{8}$}& {-$\nicefrac{7}{8}$} & {-$\nicefrac{7}{8}$} & {-$\nicefrac{7}{8}$} & {$\nicefrac{1}{8}$} & {$\nicefrac{1}{8}$} & {$\nicefrac{1}{8}$} & {$\nicefrac{1}{8}$}\\
		\bottomrule
	\end{tabular}
\vskip -0.2cm
\caption{Example with different induced notions of sparsity of all considered FTs.}
\label{tab:example}
\end{sc}
\end{table}
This bidder exhibits complementary effects for each bundle containing more than one item, as can be seen, e.g., from $3 = \pvi{110} > \pvi{100} + \pvi{010} = 2$ and $5=\pvi{111} > \pvi{100} + \pvi{010} +  \pvi{001}=3$. Observe that while this value function is sparse in FT4, i.e., $\hvic(110) = \hvic(101) = \hvic(011) = 0$, it is neither sparse in FT3 nor WHT. Note that the coefficients $\hvic(100)$, $\hvic(010)$, and $\hvic(001)$ capture the value of single items and thus cannot be zero.

The induced \emph{spectral energy} distributions for each FT, i.e., for each cardinality (i.e., number of items) $d$ from $0$ to $m=3$, we compute $\sum_{y\in \X: |y| = d} \hvic(y)^2/\sum_{y\in \X} \hvic(y)^2$, are shown in \Cref{tab:example_spectralenergy}.

\begin{table}[h!]
\centering
\begin{sc}
\scriptsize
	\begin{tabular}{
		    c
			S[table-format=2.2]
			S[table-format=2.2]
			S[table-format=2.2]
			S[table-format=2.2]
			}
			\toprule
		     & {$d=0$} & {$d=1$} & {$d=2$} & {$d=3$}\\
		     \midrule
		{FT3}& 00.00& 42.86& 42.86 & 14.28 \\
		{FT4}& 65.79& 31.58& 00.00 & 02.63 \\
		{WHT}& 65.69& 33.41& 00.68& 00.22\\
		\bottomrule
	\end{tabular}
\vskip -0.2cm
\caption{Spectral energy in $\%$ for each cardinality (i.e., number of items) $d$ from $0$ to $m=3$ of all considered FTs.}
\label{tab:example_spectralenergy}
\end{sc}
\vskip -0.1cm
\end{table}
\end{example}
\mypar{Fourier-sparse approximations} 
In practice, $\pvi{}$ may only be approximately sparse. Meaning that while not being sparse, it can be approximated well by a Fourier-sparse function $\tvi{}$. Formally, let  $\Si{1} = \supp(\htvi)$ with $|\Si{1}| = k$, we call 
\begin{equation}
\tvi{x} = \sum_{y \in \Si{1}} F^{-1}_{x,y} \htvi(y) \text{ for all } x \in \X
\end{equation}
such that $\|\tvi{} - \pvi{}\|_2$ is small a \textit{$k$-Fourier-sparse approximation} of $\pvi{}$.
We denote the vector of FCs by ${\htvSic = (\htvi(y))_{y \in \Si{1}}}$. 

\section{Fourier Transform-based WDPs}\label{FourierTransformBasedWDPsChapter}

To leverage Fourier analysis for CA design, we represent bidders' value functions using Fourier-sparse approximations.  A key step in most auction designs is to find the social welfare-maximizing allocation given bidder's reports, which is known as the \emph{Winner Determination Problem} (WDP). To apply FTs, we need to be able to solve the WDP efficiently. Accordingly, we next derive MIPs for each of the FTs. 
 
For each bidder $i \in N$, let $\tvi{}: \X \to \R_+$ be a Fourier-sparse approximation of the bidders' reported value function $\pvi{}$. Next, we define the \textit{Fourier transform-based WDP}.
\begin{definition}{(\textsc{Fourier transform-based WDP})}
\begin{align}\label{FourierTransformBasedWDP}
&\argmax\limits_{a \in \F}\sum_{i \in N}\tvi{a_i}\tag{FT-WDP}.
\end{align}
\end{definition}
For $x,y \in \R^d,$ let $x\leq y,\, \max(x,y)$ and $\left(-1\right)^x$ be defined component-wise, and let $\langle \cdot,\cdot \rangle$ denote the Euclidean scalar product. First, we formulate succinct representations of $\tvi{}$.
\begin{lemma}\label{SuccinctRepresentations}
For $i\in N$ let $\Si{1} = \{y^{(1)}, \dots, y^{(k)}\}$ be the support of a $k$-Fourier-sparse approximation $\tvi{}$ and $W_i \in \{0, 1\}^{k \times m}$ be defined as $\left(W_i\right)_{l,j} = \mathbb{I}_{y^{(l)}_j=1}$. Then it holds that
\begin{align}
\textrm{\textbf{FT3:}}\hspace{0.2cm}&\tvi{x} = \left\langle \htvSic, \max\left(0_k, 1_k\hspace{-0.07cm} -\hspace{-0.06cm} W_i(1_m\hspace{-0.06cm}-\hspace{-0.06cm}x)\right) \right\rangle\label{FT3}\\
\textrm{\textbf{FT4:}}\hspace{0.2cm}&\tvi{x} = \left\langle \htvSic, \max\left(0_k, 1_k - W_i x\right) \right\rangle\label{FT4}\\
\textrm{\textbf{WHT:}}\hspace{0.2cm}&\tvi{x} = \left\langle \htvSic, (-1)^{W_i x}\right\rangle \label{WHT}.
\end{align}
\end{lemma}
See \Cref{subsec:appendix_Proof_Lemma_1} for the proof. With Lemma \ref{SuccinctRepresentations} and rewriting $\max(\cdot,\cdot)$ and $(-1)^{\cdot}$ as linear constraints, we next encode \eqref{FourierTransformBasedWDP} as a MIP (see \Cref{subsec:appendix_Proof_Thm1} for the proof).
\begin{theorem}{\textsc{(FT-based MIPs)}}\label{FourierTransformBasedMIPs}
Let $\tvi{}:\X\to \R$ be a k-Fourier-sparse approximation from \eqref{FT3}, \eqref{FT4}, or \eqref{WHT}. Then there exists a $C>0$ s.t. the MIP defined by the objective
\begin{align}\label{appendix_GenericMIP}
&\argmax\limits_{a \in \F, \beta_i\in \{0,1\}^{k}}\sum_{i \in N}\langle\htvSic,\alpha_i\rangle,
\end{align}
and for $i\in N$ one set of transform specific constraints \eqref{FT3Contraints1}--\eqref{FT3Constraints3}, or \eqref{FT4Contraints1}--\eqref{FT4Constraints3}, or \eqref{WHTConstraints1}--\eqref{WHTConstraints3}, is equivalent to \eqref{FourierTransformBasedWDP}.
\begin{align}
\textrm{\textbf{FT3:}}\hspace{0.5cm}\textrm{s.t.}\hspace{0.15cm}&\alpha_i\ge 1_k - W_i(1_m-a_i)\label{FT3Contraints1}\\
&\alpha_i\le 1_k-W_i(1_m-a_i)+C\beta_i\label{FT3Contraints2}\\
&0_k\le \alpha_i\le C(1_k-\beta_i)\label{FT3Constraints3}\\
\textrm{\textbf{FT4:}}\hspace{0.5cm}\textrm{s.t.}\hspace{0.15cm}&\alpha_i\ge 1_k - W_i a_i\label{FT4Contraints1}\\
&\alpha_i\le 1_k-W_i a_i+C\beta_i\label{FT4Contraints2}\\
&0_k\le \alpha_i\le C(1_k-\beta_i)\label{FT4Constraints3}\\
\textrm{\textbf{WHT:}}\hspace{0.5cm}\textrm{s.t.}\hspace{0.15cm}&\alpha_i=-2\beta_i+1_k\label{WHTConstraints1}\\
& \beta_i=W_i a_i-2\gamma_i\label{WHTConstraints2}\\
&\gamma_i \in \Z^k\label{WHTConstraints3}
\end{align}
\end{theorem}
\section{Analyzing the Potential of a FT-based CA}\label{AnalyzingThePotentialOfA FourierTransformBasedCA}
In this section, we experimentally evaluate the FTs and propose an FT-based allocation rule that motivates our practical hybrid ICA mechanism presented later in Section~\ref{APracticalHybridCAMechanism}.

For our experiments, we use the spectrum auction test suite (SATS) \cite{weiss2017sats}.\footnote{We used SATS version 0.6.4 for our experiments. The implementations of GSVM and LSVM have changed slightly in newer SATS versions. This must be considered when comparing the performance of different mechanisms in those domains.} SATS enables us to generate synthetic CA instances in different domains. We have access to each bidder's \textit{true} full value function $v_i$ and the efficient allocation $a^*$. When simulating bidders, we assume truthful bidding (i.e., $\hat{v}_i=v_i$). We consider three domains:

\textbf{The Global Synergy Value Model (GSVM)} \cite{goeree2010hierarchical} has $18$ items, $6$ \textit{regional} and $1$ \textit{national bidder}.

\textbf{The Local Synergy Value Model (LSVM)} \cite{scheffel2012impact} consists of $18$ items, $5$ \textit{regional} and $1$ \textit{national bidder}. Complementarities arise from spatial proximity of items.

\textbf{The Multi-Region Value Model (MRVM)} \cite{weiss2017sats} has $98$ items and $10$ bidders (categorized as  \textit{local}, \textit{regional}, or \textit{national}) and models large US spectrum auctions.

\setlength{\textfloatsep}{\textfloatsepsave}
\subsection{Notions of Fourier Sparsity}\label{Fourier Sparsity of LSVM}

\begin{figure}
        \centering
        \includegraphics[width=1\columnwidth,trim= 15 15 10 10 ,clip]{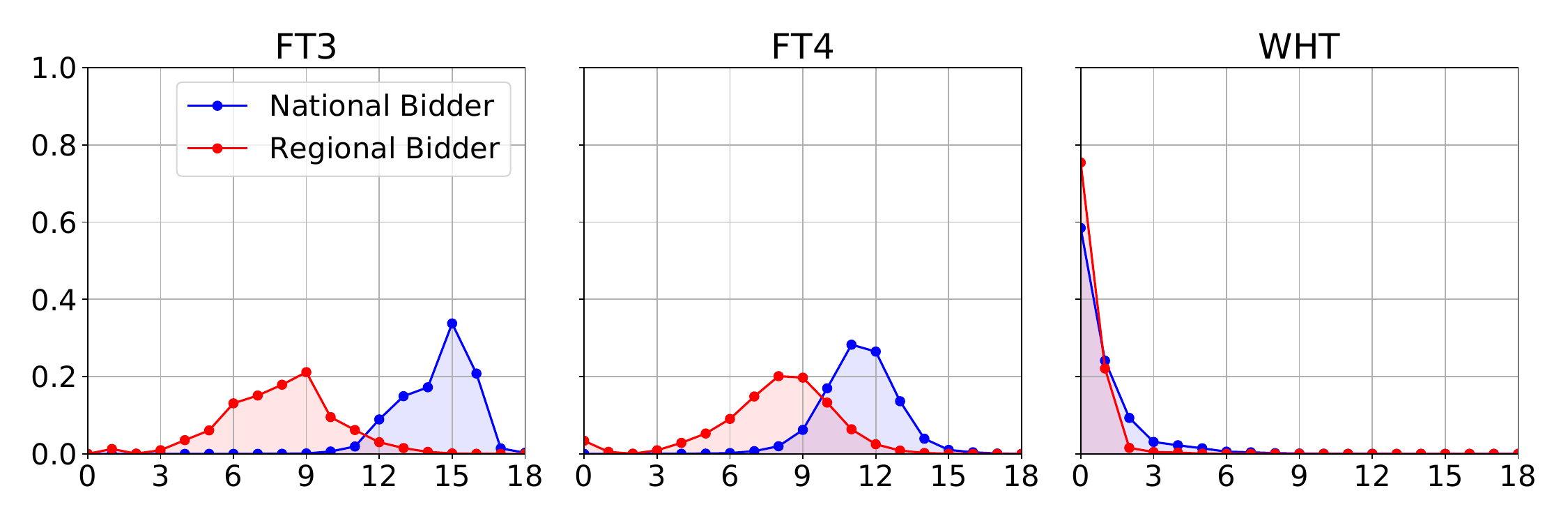}
        \vskip -0.2cm
        \caption{Spectral energy distribution in LSVM for all FTs. For each cardinality (x-axis), we plot the spectral energy of all frequencies of that cardinality normalized by the total spectral energy (y-axis).}
        \label{fig:energy_distribution}
\end{figure}

We first experimentally show that different notions of FT lead to different types of sparsity in LSVM (for other domains see \Cref{subsec:appendix_fourier_sparsity}). For this we first compute the FTs of all bidders and then calculate their \emph{spectral energy} distribution. That is, for each cardinality $d$ (\#items) from $0$ to $m$, we compute $\sum_{y\in \X: |y| = d} \hvic(y)^2/\sum_{y\in \X} \hvic(y)^2$. In Figure~\ref{fig:energy_distribution}, we present the mean over 30 LSVM instances and bidder types.

Figure~\ref{fig:energy_distribution} shows that while the energy is spread among FCs of various degrees in FT3 and FT4, in WHT the low degree ($\leq 2$) FCs contain most of the energy, i.e., the WHT has much fewer dominant FCs that accurately describe each value function. As the WHT is orthogonal, learning low degree WHT-sparse approximations leads to low reconstruction error. Low degree WHT-sparse approximations can be learnt efficiently and accurately from a small number of queries using compressive sensing~\cite{stobbe2012learning}. 

Note that the FT3 is identical to the classical polynomial value function representation~\cite{lahaie2010kernel} defined as
\begin{align}\label{eq:polyV}
\hat{v}_i^{\text{poly}}(x)=\sum_{l=1}^m\sum_{\boldsymbol{j}=\{j_1,\ldots,j_l\}\subseteq M}x_{j_1}\cdot...\cdot x_{j_l}\cdot c^{(i)}_{\boldsymbol{j}}.
\end{align}
where the coefficient $c^{(i)}_{\boldsymbol{j}}$ is equal to the FT3 FC at frequency $y$ with $y_{i} = 1$ for $i \in \{j_1, \dots, j_l\}$ and $y_i = 0$ else.\footnote{This can be seen by calculating the inverse in \eqref{FT3formula}, i.e., $F^{-1}_{y,x}=\mathbb{I}_{\min(x,y) = x}$, and plugin $F^{-1}_{y,x}$ into \eqref{eq:FTINVformula}.}
 E.g. for $M=\{1,2\},$ $\hat{v}_i^{\text{poly}}(x)=x_1c^{(i)}_{\{1\}}+x_2c^{(i)}_{\{2\}}+x_1x_2c^{(i)}_{\{1,2\}}$. Thus, converting $\hat{v}_i^{\text{poly}}$ into another FT basis (here WHT) can indeed be very helpful for the design of ML-based CAs.

\subsection{Reconstruction Error of Fourier Transforms}\label{LearningPerformance}
Next we validate the FT approach by comparing the reconstruction error of the FTs in the medium-sized GSVM and LSVM, where we can still compute the full FT (in contrast to MRVM). For now, we assume that we have access to bidders' full $\pvi{}$. In Procedure~\ref{Best FCs Full Information Setting}, we determine the best $k$-Fourier-sparse approximation $\tvi{}$ (see \Cref{subsec:appendix_reconstruction_error} for details).

\begin{algoprocedure}{\textsc{(Best FCs Given Full Access to $\pvi{}$)}}\label{Best FCs Full Information Setting}\ \\
Compute $\tvi{}$ using the $k$ absolutely largest FCs $\hvSic$ from the full FT for each bidders' reported value function $\hvic=F \pvi{}$.
\end{algoprocedure}
\begin{remark} 
Since the WHT is orthogonal and the simulated auction data is noise-free, its approximation error is exactly equal to the residual of the FCs. Thus, Procedure~\ref{Best FCs Full Information Setting} is optimal for the WHT. This is not the case for FT3 and FT4 because they are not orthogonal.
\end{remark}
\hyphenation{RMSE}
\hyphenation{RMSEs}
\begin{table}[t!]
	\robustify\bfseries
	\centering
	\begin{sc}
	\resizebox{1\columnwidth}{!}{
		\begin{tabular}{
				c
				c
				c
				S[table-format=3.1,table-figures-uncertainty=2, detect-weight]
				S[table-format=3.1,table-figures-uncertainty=2, detect-weight]
				S[table-format=3.1,table-figures-uncertainty=2, detect-weight]
				S[table-format=3.1,table-figures-uncertainty=2, detect-weight]
			}
			\toprule
			{Domain} & {k} & {Bidder} & {FT3} & {FT4} & {WHT} & {NN}\\ 
			\midrule
			\multirow{4}{*}{GSVM}&\multirow{2}{*}{100} & {Nat.} & 11.3(07)  & 14.2(08) &  \ccell 1.8(01) & 9.0(18)\\
			& & {Reg.} & \ccell0.0(00) & 1.4(02) & 0.4(01) & 7.2(09)\\ 
			\cmidrule{2-7}
			&\multirow{2}{*}{200} & {Nat.} & \ccell0.0(00)  & \ccell0.0(00) & \ccell0.0(00) & 5.7 (04)\\ 
			& & {Reg.} &\ccell0.0(00) &\ccell0.0(00) &\ccell0.0(00) & 5.2(08)\\
			\midrule
			\multirow{4}{*}{LSVM}&\multirow{2}{*}{100} & {Nat.} & 78.4(10) & 580.2(79) & \ccell31.2(04) & 48.7(12)\\  
			& & {Reg.} & 28.2(23) & 48.5(27) & \ccell6.8(03) & 17.8(07)\\
			\cmidrule{2-7}
			&\multirow{2}{*}{200} & {Nat.} & 95.8(12) & 639.0(100) & \ccell26.2(03) & 40.6(07)\\ 
			& & {Reg.} & 25.8(20) & 43.1(24) &\ccell5.3(03) & 15.3(09)\\
			\bottomrule
	\end{tabular}}
	\vskip -0.2cm
	\caption{Reconstruction error with a $95\%$-CI of $k$-Fourier-sparse approximations $\protect\tvi{}$ and NNs trained on $k$ randomly selected bundles. Winners are marked in grey.
	}
	\label{tab:LearnPerfOfFTs}
	\end{sc}
\end{table}
We then calculate the RMSE ${\scriptsize(\frac{1}{2^m}\sum_{x\in \X}(\pvi{}(x)\hspace{-0.05cm}-\hspace{-0.05cm}\tvi{x})^2)^{\nicefrac{1}{2}}}$ averaged over 100 instances and bidder types. In Table~\ref{tab:LearnPerfOfFTs}, we present the RMSEs for the three FTs and for NNs, where we used the architectures from \cited{weissteiner2020deep}.

For GSVM, we observe that we can perfectly reconstruct $\pvi{}$ with the $200$ best FCs, which shows that GSVM is 200-sparse. 
In contrast, LSVM is non-sparse, and we do not achieve perfect reconstruction with $200$ FCs. Overall, we observe that the WHT outperforms FT3 and FT4. Moreover, we see that, if we could compute the $k$ best FCs of the WHT from $k$ training points, the WHT would outperform the NNs.

However, in practice, we do not have access to full value functions. Instead, we must use an algorithm that computes the best FCs using a reasonable number of value queries. 
\begin{remark} Thanks to its orthogonality the WHT has strong theoretical guarantees for sparse recovery from few samples using compressive sensing (see \cite{stobbe2012learning}). Thus, we focus on the WHT in the remainder of this paper.
\end{remark}

\subsection{A Fourier Transform-based Allocation Rule}\label{AFourierTransformbasedAllocationRule}
We now present an FT-based allocation rule using the \textit{robust sparse WHT algorithm (RWHT)} by \cited{amrollahi2019efficiently}. RWHT learns a Fourier-sparse approximation $\tvi{}$ of $\pvi{}$ from \textit{value queries}. Procedure~\ref{FT-based Allocation Rule} finds the allocation $\tilde{a}$.
\begin{algoprocedure}\textsc{(WHT-based Allocation Rule)}\,\label{FT-based Allocation Rule}

\hspace{-0.2cm}\textbf{i.} Use RWHT to compute $k$-sparse approximations $\tvi{}\,,i\in N$.

\hspace{-0.2cm}\textbf{ii.} Solve ${\tilde{a}\hspace{-0.03cm}\in\hspace{-0.03cm} \argmax\limits_{a \in \F}\hspace{-0.08cm}\sum\limits_{i \in N}\hspace{-0.08cm}\tvi{a_i}}$ using Theorem~\ref{FourierTransformBasedMIPs}.
\end{algoprocedure}

In Figure~\ref{fig:FTBasedCAMechanism}, we present the efficiency of $\tilde{a}$ on 50 GSVM instances for various values of $k$. We see that RWHT achieves a median efficiency of 100\% for $90$ or more FCs. Nevertheless, the main practical issue with this approach is the number of value queries required. As we can see, RWHT needs $102,000$ value queries (39\% of all bundles) to find the $90$ best FCs. For a practical ICA mechanism this is far too many.
\begin{figure}[h!]
        \centering
        \includegraphics[width=1\columnwidth]{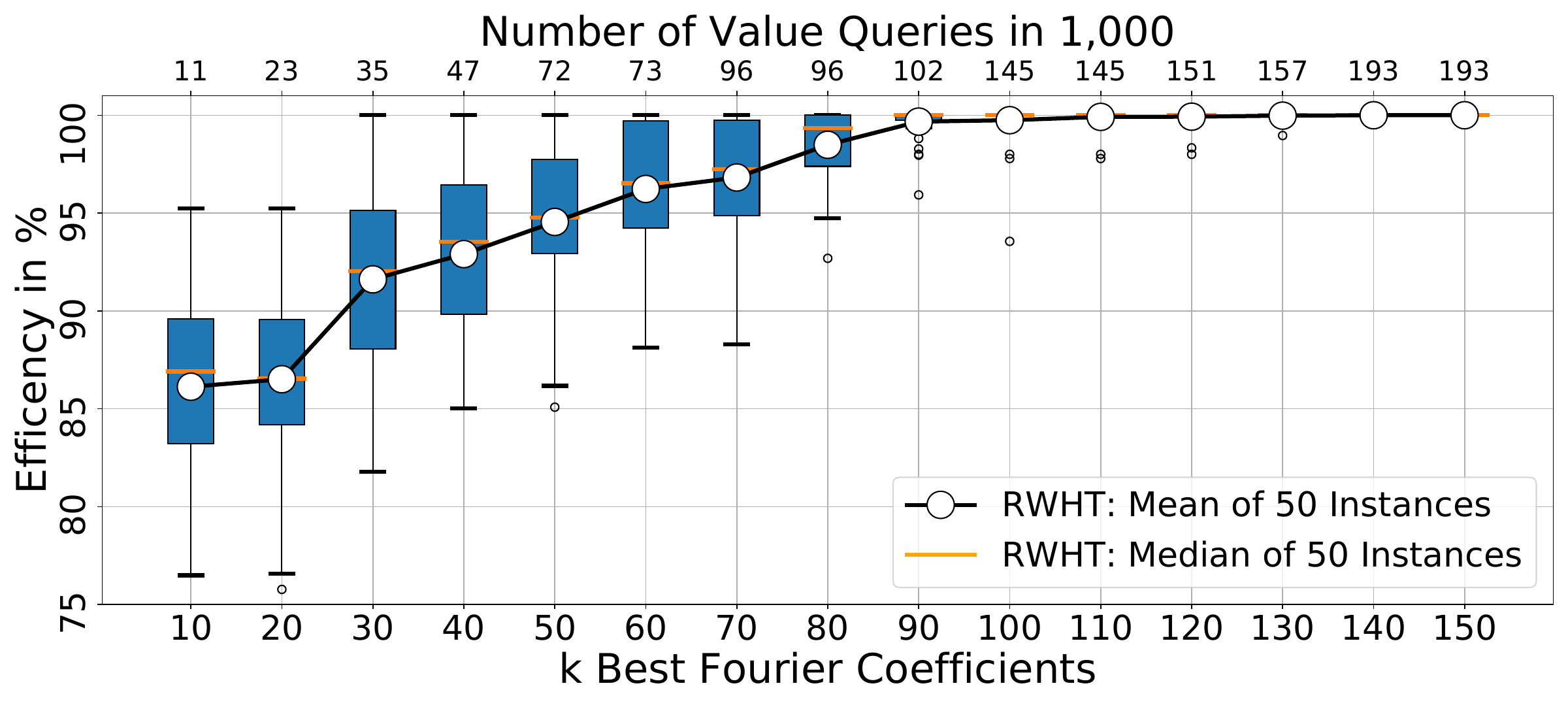} 
        \vskip -0.2cm
        \caption{Efficiency of Procedure 2 in GSVM.}
        \label{fig:FTBasedCAMechanism}
\end{figure}
\section{A Practical Hybrid ICA Mechanism}\label{APracticalHybridCAMechanism}
In this section, we introduce and experimentally evaluate a practical hybrid ICA mechanism, based on FTs \textit{and} NNs. 

\subsection{The Hybrid ICA Mechanism}\label{TheHybridMechanism}
The main issue of the FT-based allocation rule in Section~\ref{AFourierTransformbasedAllocationRule} is the large number of queries, which we now address. The idea is the following: instead of directly applying a sparse FT algorithm (like RWHT) to bidders, we apply it to a NN-based representation. In this way, we query NNs instead of bidders. Based on the FCs of the NNs, we determine a Fourier-sparse approximation $\tvi{}$ with only \emph{few value queries}, where the idea is that the FCs of each NN concentrate on the most dominant FCs of its respective value function. Indeed, recent evidence suggests that a NN trained by SGD can learn the Fourier-support~\cite{rahaman2019spectral}. We analyze our NN support discovery rule in Section~\ref{NNs Support Discovery Experiments}. We now present \textsc{Hybrid Ica}, leaving details of the sub-procedures to \Cref{subsec:appendix_DetailsOfHybridICA}.

\textsc{Hybrid Ica} (Algorithm \ref{HybridICA}) consists of 3 phases: the \emph{MLCA}, the \emph{Fourier reconstruction}, and the \emph{Fourier allocation phase.} It is parameterized by an FT $F$ and the numbers $\pr,\pml,\pfr,\pfa$ of different query types. In total, it  asks each bidder $\sum_{i=1}^4\ell_i$ queries: $\pr$ random initial, $\pml$ MLCA, $\pfr$ Fourier reconstruction, and $\pfa$ Fourier allocation queries.

\mypar{1. MLCA Phase} We first run \textsc{Mlca} such that the NNs can then be trained on ``meaningfully'' elicited reports. In \textsc{Mlca}, we request reports for $\pr$ random initial bundles and for $\pml$ \textit{MLCA queries} (Lines 1-2).
\setlength{\textfloatsep}{5pt}
\begin{algorithm}[t!]
        \DontPrintSemicolon
        \SetKwInOut{parameters}{Params}
        \parameters{$F$ Fourier transform; $\pr,\pml,\pfr,\pfa$ query split}
    Set $\Qinit= \pr$ and $\Qmax= \pr+\pml$\Comment*[r]{\color{blue}{MLCA phase}}
    Run \textsc{Mlca}($\Qinit,\Qmax$); get $\pr + \pml$ reports $R$\; 
    \ForEach(\Comment*[f]{\color{blue}{Fourier reconstr. phase}}){bidder $i\in N$}{
    Fit NN $\Ni{}$ to $R_i$\;
    Determine the best FCs of $\Ni{}$ \Comment*[r]{\textbf{Proc.~3}}
    Compute $\pfr$ reconstruction queries $\Si{2}\subseteq \X$ \Comment*[r]{\textbf{Proc.~4}}
    Ask $\Si{2}$, add reports to $R_i$, and fit $\tvi{}$ to $R_i$\Comment*[r]{\textbf{Proc.~5}} 
    }
    \For(\Comment*[f]{\textcolor{blue}{Fourier alloc. phase\hspace{0.25cm}}}) {$l=1,\ldots,\pfa$}{
    Solve $q\in \argmax_{a \in \F}\sum_{i \in N}\tvi{a_i}\quad \eqref{FourierTransformBasedWDP}$\;
    \ForEach{bidder $i \in N$}{
    \If(\Comment*[f]{Bundle already queried}){$q_i\in R_i$}{ 
        Define $\pF=  \{a\in \F : a_i \neq x, \forall x\in R_i\}$\;
        Resolve $\pq \in \argmax_{a \in \pF}\sum_{i \in N} \tilde{v}_{i}(a_{i})$\;
        Update $q_i = \pqi\;$
        }
    Query bidder $i$'s value for $q_i$ and add report to $R_i$\;
    Fit $\tvi{}$ to $R_i$ \Comment*[r]{\textbf{Proc.~5}} 
    }
    }
    From $R$ compute: $a^*_{R}$ as in \cref{WDPFiniteReports}, VCG payments $p(R)$ \;
    \Return{Final allocation $a^*_{R}$ and VCG payments $p(R)$}
    \caption{\small \textsc{Hybrid Ica}}
    \label{HybridICA}
\end{algorithm}

\mypar{2. Fourier Reconstruction Phase} Next, we compute a Fourier-sparse approximation $\tvi{}$. For this, we first fit a NN $\Ni{}$ on the reports $R_i$ (Line 4). Then we compute the best FCs of the fitted NNs (Line 5, \textbf{Procedure 3}) in order to discover which FCs are important to represent the bidders. Based on these FCs, we determine $\pfr$ \textit{Fourier reconstruction queries} $\Si{2}$ (Line 6, \textbf{Procedure 4}), send them to the bidders and fit $\tvi{}$ to the reports $R_i$ received so far (Line 7, \textbf{Procedure 5}).

\mypar{3. Fourier Allocation Phase} We use the fitted $\tvi{}$ to generate $\pfa$ \textit{Fourier allocation queries}. Here, we solve the FT-based WDP (Line 9) to get candidate queries $q$, ensure that all queries are new (Lines 11--14), add the received reports to $R_i$ (Line 15) and refit $\tvi{}$ (Line 16). Finally in Line 17, \textsc{Hybrid Ica} computes based on all reports $R$ a welfare-maximizing allocation $a^*_R$ and VCG payments $p(R)$ (see \Cref{subsec:appendix_VCGpayments}).

\mypar{Experiment Setup} For \textsc{Hybrid Ica} and \textsc{Mlca}, we use the NN architectures from \cited{weissteiner2020deep} and set a total query budget of $100$ (GSVM, LSVM) and $500$ (MRVM). For \textsc{Hybrid Ica}, we optimized the FTs and query parameters $\ell_i$ on a training set of CA instances. Table~\ref{tab:HybridICABestConfig} shows the best configurations.
\begin{table}[h!]
\centering
\begin{sc}
\resizebox{1\columnwidth}{!}{
	\begin{tabular}{llccccc}
	    \toprule
		& NN Architectures  & FT & $\pr$ & $\pml$ & $\pfr$ & $\pfa$\\
		\midrule
		GSVM & R:$[32,32]\,|$ N:$[10,10]$ & WHT & $30$ & $21$ & $20$ & $29$\\
		LSVM & R:$[32,32]\, |$ N:$[10,10,10]$ & WHT & $30$ & $30$ & $10$ & $30$\\
		MRVM & L,R,N:$[16,16]$ & WHT & $30$ & $220$ & $0$ & $250$\\
		\bottomrule
	\end{tabular}
}
\vskip -0.2cm
\caption{Best configuration of \textsc{Hybrid Ica}. R:$[d_1,d_2]$ denotes a 3-hidden-layer NN for the regional bidder with $d_1$, and $d_2$ nodes.}
\label{tab:HybridICABestConfig}
\end{sc}
\end{table}
\begin{table*}[ht]
	\renewcommand\arraystretch{1.2}
	\centering
	\begin{sc}
	\resizebox{\textwidth}{!}{%
		\begin{tabular}{
				l
				S[table-format=2.2,table-figures-uncertainty=2, detect-weight]
				S[table-format=2.2]
				S[table-format=1.2]
				S[table-format=2.0]
				S[table-format=1.2]
				S[table-format=3.2,table-figures-uncertainty=2, detect-weight]
				S[table-format=2.2]
				S[table-format=1.2]
				S[table-format=2.0]
				S[table-format=1.2]
				S[table-format=2.2,table-figures-uncertainty=2, detect-weight]
				S[table-format=1.2]
				S[table-format=1.2]
				S[table-format=2.2]
				S[table-format=2.0]
				S[table-format=2.2]
			}
			\toprule
			& \multicolumn{5}{c}{\textbf{GSVM}}               & \multicolumn{5}{c}{\textbf{LSVM}}                 & \multicolumn{6}{c}{\textbf{MRVM}}\\
			\cmidrule(lr){2-6}
			\cmidrule(lr){7-11}
			\cmidrule(lr){12-17}
			& {\textbf{Efficiency}} &  {\textbf{Regional}}  & {\textbf{National}}  &  {\textbf{Rev}} & {\textbf{hrs/}}  &  {\textbf{Efficiency}} & {\textbf{Regional}}  & {\textbf{National}}  & {\textbf{Rev}}  & {\textbf{hrs/}}   &  {\textbf{Efficiency}} & {\textbf{Local}}   & {\textbf{Regional}}  & {\textbf{National}}  & {\textbf{Rev}}  & {\textbf{hrs/}} \\
			\textbf{Mechanism} & {in \textbf{\%}}      &  {in \textbf{\%}} & {in \textbf{\%}} & {in \textbf{\%}} & {\textbf{Inst.}} & { in \textbf{\%}}      & {in \textbf{\%}} & {in \textbf{\%} }&{ in \textbf{\%}} & {\textbf{Inst.}} & {in \textbf{\%} }      & {in \textbf{\%}} & {in \textbf{\%}} & {in \textbf{\%}} & {in \textbf{\%} }&{ \textbf{Inst.}} \\ 
			\midrule
			\textsc{Hybrid Ica} & \ccell99.97(3)  & 94.72  & 5.25 & 81 & 0.81 & \ccell98.74(43) & 89.09 & 9.65 & 78 & 1.95 & \ccell96.63(31) & 0.00 & 1.19 & 95.44 & 36 & 23.88 \\
			\textsc{Mlca}       & 99.17(37)  & 98.11  & 1.06 & 79 & 4.65 & \ccell99.14(42) & 93.40 & 5.75 & 77 & 6.09 & 95.32(32) & 0.00 & 0.53 & 94.79 & 41 & 43.26 \\
			\textsc{Hybrid Ica (no FR)}& 98.30(49)  & 97.94  & 0.36 & 75 & 0.93 & \ccell98.16(60) & 93.83 & 4.33 & 72 &  2.03& \ccell96.63(31) & 0.00 & 1.19 & 95.44 & 36 & 23.88 \\
			\textsc{Hybrid Ica (no FR/FA)}& 98.16(50)        & 97.47  & 0.69 & 75 & 0.71 & 97.75(63)       & 92.78 & 5.27 & 72 & 1.86 & 93.91(36)       & 0.01 & 0.42 & 93.48 & 42 & 14.68 \\
			\midrule
			{\textit{Efficient Allocation}}&            & 94.75  & 5.25 &    &      &           & 84.03 & 15.97&    &      &           & 0.00 & 2.11 & 97.89 &    & \\
			\bottomrule 
		\end{tabular}
	}%
	\vskip -0.2cm
	\caption{\textsc{Hybrid Ica} vs. \textsc{Mlca}, \textsc{Hybrid Ica (no FR)}, and \textsc{Hybrid Ica (no FR/FA)}. All results are averages over a test set of 100 (GSVM and LSVM) and 30 (MRVM) CA instances. For efficiency we give a 95\% confidence interval and mark the best mechanisms in grey.}
	\label{tab:HybridICAResultsSummary}
	\end{sc}
\end{table*}%

\subsection{NNs Support Discovery Experiments}\label{NNs Support Discovery Experiments}
In \textsc{Hybrid ICA} we use the NNs for support discovery where it is key that the FCs of these NNs concentrate on the dominant FCs of its value function, i.e. $\supp(\hNi)\approx\supp(\hvic)$. 

To evaluate the NN-based support discovery (Line 5), we consider the spectral \textit{energy ratio} obtained by dividing the spectral energies of the $k$ frequencies selected from the NN and the $k$ best frequencies (for the WHT the best FCs are the ones with the largest absolute value). Formally, for each bidder $i$, let the $k$ frequencies selected from the NN be $\mathcal{\tilde{S}}_i = \{\tilde{y}^{(1)}, \dots, \tilde{y}^{(k)}\}$ and the best ones be $\mathcal{S}_i^{*} = \{\overset{*}{y}{}^{(1)}, \dots, \overset{*}{y}{}^{(k)}\}$. 
Then, bidder $i$'s \textit{energy ratio} is given by $\sum_{\tilde{y} \in \tilde{S}_i} \hvic(\tilde{y})^2/\sum_{\overset{*}{y}{}\in \mathcal{S}_i^{*}} \hvic(\overset{*}{y}{})^2\in [0,1]$ (see \Cref{subsec:appendix_NN_support_discovery} for details).
This ratio is equal to one if $\tilde{S}_i=\mathcal{S}_i^{*}$.
Figure~\ref{fig:energy_ratio} shows that the NN-based supports are almost on par with the best supports given a fixed budget of $k$ frequencies.
\begin{figure}[h!]
        \centering
        \includegraphics[width=1\columnwidth]{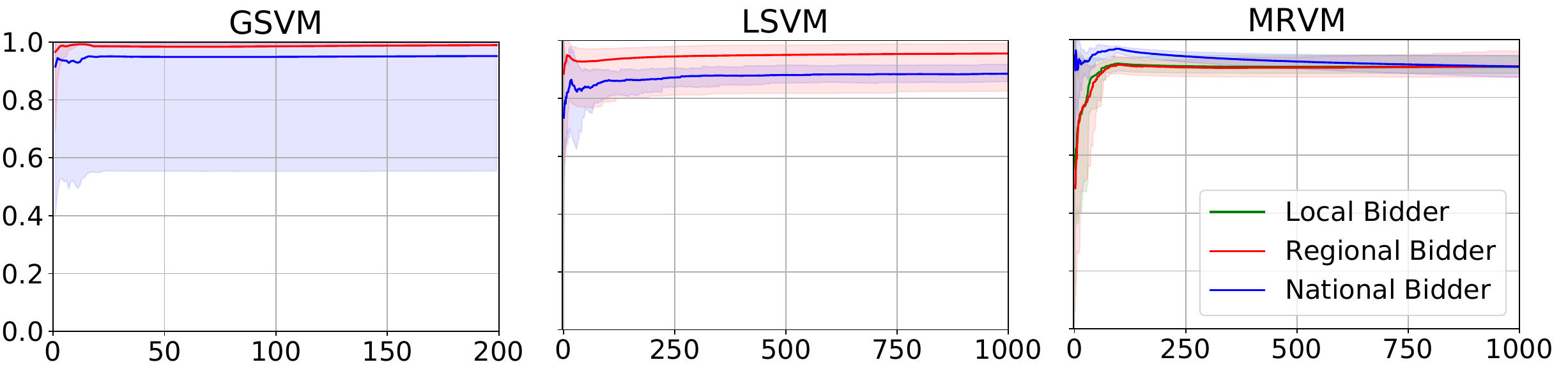}
        \vskip -0.2cm
        \caption{Average energy ratio (y-axis) with 97.5\% and 2.5\% empirical quantiles for a number of selected frequencies $k$ (x-axis) over 30 instances in GSVM and LSVM and over 5 instances in MRVM.}
        \label{fig:energy_ratio}
\end{figure}
\subsection{Efficiency Experiments}\label{Experiments}
We now evaluate the efficiency of \textsc{Hybrid Ica} vs \textsc{Mlca}.

\mypar{Results} Table~\ref{tab:HybridICAResultsSummary} contains our main results in all domains.\footnote{All experiments were conducted on machines with Intel Xeon E5 v4 2.20GHz processors with 24 cores and 128GB RAM or with Intel E5 v2 2.80GHz processors with 20 cores and 128GB RAM.} We show efficiency, distribution of efficiency to bidder types, revenue ${\scriptsize(\sum_{i \in N}p(R)_i)/V(a^*)}$, and runtime. First, we see that \textsc{Hybrid Ica} \emph{statistically significantly outperforms} \textsc{Mlca} \emph{w.r.t. efficiency} in GSVM and MRVM and performs on par in LSVM. Second, it also leads to a computational speedup ($\times6$ GSVM, $\times3$ LSVM, $\times2$ MRVM). The reason for this computational speedup is that the generation of
the $\ell_3 + \ell_4$ Fourier queries (estimating the superset of the support using RWHT on the NNs, fitting the FT models using
compressive sensing and solving our new FT-based MIPs) is
faster than the generation of the NN-based MLCA allocation
queries (training NNs and solving the NN-based MIP).
Third, it distributes the welfare more evenly (= fairer) to bidder types.\footnote{We consider an allocation to be ``fairer'' if its social welfare
is more evenly distributed among bidder types. This is similar (but not identical) to the standard notion of egalitarian
social welfare.} This also leads to a distribution that more closely resembles that of the efficient allocation (see \textit{Efficient Allocation}). We present full efficiency path plots for the different phases of \textsc{Hybrid Ica} in \Cref{subsec:appendix_efficiency_experiments}.

\mypar{Fourier queries}  To verify the importance of the $\pfr$ \textit{Fourier reconstruction} and $\pfa$ \textit{Fourier allocation queries}, we also present \textsc{Hybrid Ica (no FR)} and \textsc{Hybrid Ica (no FR/FA)}, which use random queries in place of the $\pfr$ Fourier reconstruction and the $\pfr+\pfa$ Fourier-based queries. As we see in Table~\ref{tab:HybridICAResultsSummary}, using the Fourier queries leads to significantly better efficiency and \textsc{Hybrid Ica (no FR/FA)} does not achieve a fairer efficiency distribution. A comparison of \textsc{Hybrid Ica} to \textsc{Hybrid Ica (no FR)} reveals that, in GSVM and LSVM, the Fourier reconstruction queries cause the fairer distribution. We empirically verified that these queries are composed of larger bundles (on avg. $17$ items vs. $4$ in MLCA queries) and thus allocate large bundles to bidders that would have been overlooked. In MRVM, the optimal query split for \textsc{Hybrid Ica} uses $\pfr=0$ Fourier reconstruction queries such that \textsc{Hybrid Ica} is equal to \textsc{Hybrid Ica (no FR)}. Thus, in MRVM, \textsc{Hybrid Ica}'s increased efficiency and fairer distribution results from the Fourier allocation queries.

Overall, we see that our Fourier-based auction is especially powerful in sparse domains. In practice, bidders are often limited by their cognitive abilities \cite{scheffel2012impact} or use a low-dimensional computational model to represent their value function. Thus, their reported preferences typically exhibit only a limited degree of substitutability and complementarity, which is captured well by Fourier-sparsity.
\section{Conclusion}\label{Conclusion}
We have introduced Fourier analysis for the design of CAs. The main idea was to represent value functions using Fourier-sparse approximations, providing us with a new lens on bidder's values in the Fourier domain.

On a technical level, we have derived succinct MIPs for the Fourier transform-based WDPs, which makes computing Fourier-based allocation queries practically feasible. We have leveraged this to design a new hybrid ICA that uses NN and Fourier-queries. Our experiments have shown that our approach leads to higher efficiency, a computational speedup and a fairer distribution of social welfare than state-of-the-art.

With this paper, we have laid the foundations for future work leveraging Fourier analysis for representing and eliciting preferences in combinatorial settings.

\section*{Acknowledgments}\label{Acknowledgments}
We thank the anonymous reviewers for helpful comments. This paper is part of a project that has received funding from the European Research Council (ERC) under the European Union’s Horizon 2020 research and innovation programme (Grant agreement No. 805542). Furthermore, this material is based upon work supported by the National Science Foundation under grant no. CMMI-1761163.

\bibliographystyle{named}
\bibliography{references_arXiv}

\clearpage
\appendix
\section*{\centering Appendix}
\section{Preliminaries}

\subsection{A Machine Learning powered ICA}\label{subsec:appendix_AnNNPoweredICAMechanism}
In this section, we provide a detailed review of the \textit{machine learning-powered combinatorial auction (MLCA)} introduced by \cited{brero2021workingpaper}. We describe MLCA using NNs $\Ni{}$ as the generic ML algorithm $\mathcal{A}_i$ (see \cite{weissteiner2020deep}).

At the core of MLCA is a \textit{query module} (Algorithm~\ref{QueryModule}), which, for each bidder $i\in I$, determines a new value query $q_i$. First, in the \textit{estimation step} (Line 1), NNs are used to learn bidders' valuations from reports $R_i$. Next, in the \textit{optimization step} (Line 2), a \textit{NN-based WDP} is solved to find a candidate $q$ of value queries (see \cite{weissteiner2020deep} for details on the NN-based \textit{estimation} and \textit{optimization step}). Finally, if $q_i$ has already been queried before (Line 4), another, more restricted NN-based WDP (Line 6) is solved and $q_i$ is updated correspondingly. This ensures that all final queries $q$ are new.
\setlength{\textfloatsep}{5pt}
\begin{algorithm}[h!]
        \DontPrintSemicolon
        \SetKwInOut{procedure}{Function}
        \SetKwInOut{parameters}{Parameters}
        \SetKwInOut{inputs}{Inputs}
        \procedure{~\textit{NextQueries}$(I,R)$}
        \inputs{~Index set of bidders $I$ and reported values $R$}
        \parameters{~Neural networks $\Ni{}:\X\to \R_+,\, i\in N$}
    \lForEach(\Comment*[f]{\color{blue}Estimation step}){$i \in I$}{
    \hspace{-0.07cm}Fit $\Ni{}$ on $R_i$: $\Ni{}[R_i]$
    }
    Solve $q \in \argmax\limits_{a \in {\F}}\sum\limits_{i \in I} \Ni{}[R_i](a_i)$ \hspace{-0.03cm}\Comment*[r]{\color{blue}Optimization step}
    \ForEach{$i \in I$}{
        \If(\Comment*[f]{\color{blue} Bundle already queried}){$q_i\in R_i$}{ 
        Define $\pF=  \{a\in \F : a_i \neq x, \forall x\in R_i\}$\;
        Re-solve $\pq \in \argmax_{a \in \pF}\sum_{l \in I} \mathcal{N}_l{}[R_l](a_l)$\;
        Update $q_i = \pqi\;$
        }
    }
    \Return{profile of new queries $q=(q_1,\ldots,q_n)$}
    \caption{\textsc{NN-Query Module}\, {\scriptsize (Brero et al. 2021)}}
    \label{QueryModule}
\end{algorithm}

\setlength{\textfloatsep}{5pt}
\begin{algorithm}[t!]
        \DontPrintSemicolon
        \SetKwInOut{parameters}{Params}
        \parameters{$t=1,\Qinit,\Qmax$ {init. and max \#queries per bidder}}
    \ForEach{$i \in N$}{Receive reports $R_i$ for $\Qinit$ randomly drawn bundles}
    \While(\Comment*[f]{\color{blue}Round iterator}){$t\leq \floor{(\Qmax-\Qinit)/n}$}{
        $\qnew=$ \textit{NextQueries$(N,R)$} \Comment*[r]{\color{blue}Main economy queries}
        \ForEach(\Comment*[f]{\color{blue}Marginal economy queries}){bidder $i \in N$}{
            $\qnew=\qnew\cup$ \textit{NextQueries$(N\setminus\{i\},R_{-i})$}
        }
        \ForEach{bidder $i \in N$}{
         Receive reports $\Rnewi$ for $\qnew_i$, set $R_i=R_i\cup\Rnewi$
        }
        $t = t+1$
    }
    From elicited reports $R$ compute $a^*_{R}$ as in \Cref{WDPFiniteReports}\;
    From elicited reports $R$ compute VCG-payments $p(R)$\;
    \Return{Final allocation $a^*_{R}$ and payments $p(R)$}
    \caption{\small \textsc{Mlca}\, {\scriptsize (Brero et al. 2021)}}
    \label{MLCA}
\end{algorithm}

In Algorithm \ref{MLCA}, we present \textsc{Mlca} in a slightly abbreviated form. Let $R_{-i}=(R_1,\ldots,R_{i-1},R_{i+1},\ldots, R_n)$. \textsc{Mlca} proceeds in rounds until a maximum number of queries per bidder $\Qmax$ is reached. In each round, it calls Algorithm \ref{QueryModule} $n+1$ times: once including all bidders (Line 4, \textit{main economy}) and $n$ times excluding one bidder (Lines 5--6, \textit{marginal economies}). At the end of each round, the mechanism receives reports $\Rnew$ from all bidders for the newly generated queries $\qnew$, and updates the overall elicited reports $R$ (Lines 7--8). In Lines 10--11, \textsc{Mlca} computes an allocation $a^*_R$ that maximizes the reported social welfare and determines VCG payments $p(R)$ (see \Cref{subsec:appendix_VCGpayments}).

\subsection{VCG Payments from Reports}\label{subsec:appendix_VCGpayments}
In this section, we provide a recap on how to compute VCG payments from bidder's reports.
Let $R=(R_1,\ldots,R_n)$ denote an elicited set of reported bundle-value pairs from each bidder obtained from \textsc{Mlca} (Algorithm~\ref{MLCA}) or \textsc{Hybrid Ica} (Algorithm~\ref{HybridICA}) and let $R_{-i}\coloneqq(R_1,\ldots,R_{i-1},R_{i+1},\ldots,R_n)$. We then calculate the VCG payments $p(R)=(p(R)_1\ldots,p(R)_n) \in \R_+^n$ as follows:
\begin{definition}{\textsc{(VCG Payments from Reports)}}
\begin{align}\label{VCGPayments}
&p(R)_i \coloneqq \hspace{-0.2cm}\sum_{j \in N \setminus \{i\}} \pvj{}\left(\left(a^*_{R_{-i}}\right)_j\right) - \hspace{-0.2cm}\sum_{j \in N \setminus \{i\}}\pvj{}\left(\left(a^*_{R}\right)_j\right).
\end{align}
where $a^*_{R_{-i}}$ maximizes the reported social welfare when excluding bidder $i$, i.e.,
\begin{align}
&a^*_{R_{-i}}\in \argmax_{a \in \F} \hV{a|R_{-i}} = \argmax_{a \in \F}\hspace{-0.4cm}\sum_{\substack{j \in N\setminus\{i\}:\\ \left(a_j,\pvj{a_j}\right)\in R_j}}\hspace{-0.4cm}\pvj{a_j},
\end{align}
and $a^*_R$ is a reported-social-welfare-maximizing allocation (including all bidders), i.e,
\begin{align}
&a^*_{R}\in \argmax_{a \in \F} \hV{a|R} = \argmax_{a \in \F} \hspace{-1.3cm}\sum_{\hspace{1cm}i \in N:\, \left(a_i,\pvi{a_i}\right)\in R_i}\hspace{-1.3cm}\pvi{a_i}.
\end{align}
\end{definition}
As argued by \cited{brero2021workingpaper}, using such payments are key for \textsc{Mlca} to induce ``good'' incentives for bidders to report truthfully. As their incentive analysis also applies to \textsc{Hybrid Ica}, we use them in our design too.

\section{Fourier Transform-based WDPs}
In this section, we present the proofs of Lemma~\ref{SuccinctRepresentations} and Theorem~\ref{FourierTransformBasedMIPs}.

Let $1_d$ and $0_d$ for $d\in \N$ denote the $d$-dimensional vector of ones and zeros, respectively. For all $x,y \in \R^d,$ let $x\leq y,\, \max(x,y),\, \min(x,y)$ and $\left(-1\right)^x$  be defined component-wise, and let $\langle \cdot,\cdot \rangle$ denote the Euclidean scalar product.

We consider the matrix representation $F$ and $F^{-1}$ of the considered FTs from \cited{puschel2020discrete} given by: 

\begin{align}
\textrm{\textbf{FT3:}}\hspace{0.5cm}&F_{y,x} = (-1)^{|y| - |x|} \mathbb{I}_{\min(x,y) = x},\label{appendix_FT3formula}&\\
 &F^{-1}_{x,y} = \mathbb{I}_{\min(x,y) = y},\label{appendix_FT3formulaINV}&\\
\textrm{\textbf{FT4:}}\hspace{0.5cm}&F_{y,x} = (-1)^{|\min(x,y)|} \mathbb{I}_{\max(x,y) = 1_m},\label{appendix_FT4formula}&\\
 &F^{-1}_{x,y} = \mathbb{I}_{\min(x,y) = 0_m},\label{appendix_FT4formulaINV}&\\
\textrm{\textbf{WHT:}}\hspace{0.5cm}&F_{y,x} = \frac{1}{2^m} (-1)^{|\min(x, y)|},\label{appendix_WHTformula}\\
 &F^{-1}_{x,y} = (-1)^{|\min(x,y)|}.\label{appendix_WHTformulaINV}
\end{align}

The Fourier-sparse approximations used in the WDPs are defined in terms of $F^{-1}$. 

\subsection{Proof of Lemma~\ref{SuccinctRepresentations}}\label{subsec:appendix_Proof_Lemma_1}
\setcounter{lemma}{0}
\begin{lemma}{\textsc{(Succinct Representations)}}
For $i\in N$ let $\Si{1} = \{y^{(1)}, \dots, y^{(k)}\}$ be the support of a $k$-Fourier-sparse approximation $\tvi{}$ and $W_i \in \{0, 1\}^{k \times m}$ be defined as $\left(W_i\right)_{l,j} = \mathbb{I}_{y^{(l)}_j=1}$. Then $\tvi{}$ can be rewritten as:
\begin{align}
\textrm{\textbf{FT3:}}\hspace{0.2cm}&\tvi{x} = \left\langle \htvSic, \max\left(0_k, 1_k\hspace{-0.07cm} -\hspace{-0.06cm} W_i(1_m\hspace{-0.06cm}-\hspace{-0.06cm}x)\right) \right\rangle\label{appendix_FT3}\\
\textrm{\textbf{FT4:}}\hspace{0.2cm}&\tvi{x} = \left\langle \htvSic, \max\left(0_k, 1_k - W_i x\right) \right\rangle\label{appendix_FT4}\\
\textrm{\textbf{WHT:}}\hspace{0.2cm}&\tvi{x} = \left\langle \htvSic, (-1)^{W_i x}\right\rangle \label{appendix_WHT}.
\end{align}


\end{lemma}
\begin{proof}
In this proof we are going to make use of the equivalence between bundles (= indicator vectors) and sets. Recall that $x, y \in \X=\{0,1\}^m$ are bundles and bundles correspond to subsets of items. We can translate set operations such as intersection, union and complement to indicator vectors as follows:
\begin{align}
&x \cap y = \min(x, y),\label{intersection}\\
&x \cup y = \max(x, y),\label{union}\\
&(x)^c = 1_m - x,\label{complement}
\end{align}
where we slightly abused notation by identifying the indicator vectors with their corresponding subsets on the left hand sides of \eqref{intersection}--\eqref{complement}. Furthermore, it holds that:
\begin{align}
|\min(x, y)| = y^T x = \sum_{i = 1}^m x_i y_i
\end{align}

\mypar{FT3} Let $F$ denote the matrix representation of FT3 in \eqref{appendix_FT3formula}. 
By definition, we have 
\begin{align}
    \tvi{x} =& \sum_{y \in \Si{1}} F^{-1}_{x,y} \htvi(y)\\ 
    =&\sum_{y \in  \Si{1}}\mathbb{I}_{\min(x, y) = y}\htvi(y)\\
    =&\sum_{y \in  \Si{1}}\mathbb{I}_{\min(1_m - x, y) = 0_m}\htvi(y)\\
    =& \langle (\mathbb{I}_{\min(1_m - x, y) = 0_m})_{y \in \Si{1}}, \htvSic \rangle,
\end{align}
where we used $x \cap y = y \Leftrightarrow x^c \cap y = \emptyset$. Now, the claim follows by observing that 
\begin{align}
\mathbb{I}_{\min(1_m - x, y) = 0_m} = \max\left(0, 1 - y^T (1_m - x\right)),
\end{align}
which is a direct consequence of $|\min(1_m - x, y)|=y^T (1_m - x)$, and recalling that for each $y \in \Si{1}$ there is a respective row $y^T$ in $W_i$.\\

\mypar{FT4} Let $F$ be the matrix representation of FT4 in \eqref{appendix_FT4formula}. By definition, we have 
\begin{align}
    \tvi{x} =& \sum_{y \in \Si{1}} F^{-1}_{x,y} \htvi(y)\\ =&\sum_{y \in \Si{1}} \mathbb{I}_{\min(x, y) = 0_m}\htvi(y)\\
    =& \langle (\mathbb{I}_{\min(x, y) = 0_m})_{y \in \Si{1}}, \htvSic \rangle.
\end{align}
Now, the claim follows by observing that 
\begin{align}
\mathbb{I}_{\min(x, y) = 0_m} = \max\left(0, 1 - y^T x\right),
\end{align}
which is a direct consequence of $|\min(x, y)|=y^T x $, and recalling that for each $y \in \Si{1}$ there is a respective row $y^T$ in $W_i$.\\

\mypar{WHT} Let $F$ be the matrix representation of WHT in \eqref{appendix_WHTformula}. By definition, we have
\begin{align}
    \tvi{x} =& \sum_{y \in \Si{1}} F^{-1}_{x,y} \htvi(y)\\
    =& \sum_{y \in \Si{1}} (-1)^{y^T x}\htvi(y)\\
    =& \langle ((-1)^{y^T x})_{y \in \Si{1}}, \htvSic \rangle.
\end{align}
Now, the claim follows by recalling that for each $y \in \Si{1}$ there is a respective row $y^T$ in $W_i$.
\end{proof}

\subsection{Proof of Theorem 1}\label{subsec:appendix_Proof_Thm1}
We first state and proof two elementary lemmata. 

For $x\in \R^d$ let $x \Mod 2$  be defined component-wise.
\begin{lemma}{\textsc{(Max Representation)}}\label{MaxAsLinearConstraints}
Let $\zeta,\eta \in \R^d$ for $d \in \N$ and $C>0$ such that for all $l\in \{1,\ldots,d\}$ it holds that $|\zeta_l-\eta_l|\le C$. Let $\alpha:= \max(\zeta,\eta)$ and consider the polytope $\mathcal{P}$ in $(\talpha,\beta)$ defined by \eqref{maxpolytopeStart}--\eqref{maxpolytopeEnd}
\begin{align}
    &\talpha\ge \eta\label{maxpolytopeStart}\\
    &\talpha\le \eta + C\beta\\
    &\talpha\ge \zeta\\
    &\talpha\le \zeta + C(1_d-\beta)\\
    &\beta\in \{0,1\}^{d}.\label{maxpolytopeEnd}
\end{align}
Then it holds that $\mathcal{P}\neq\emptyset$ and every element $(\talpha,\beta)\in \mathcal{P}$ satisfies $\talpha=\alpha.$
\end{lemma}
\begin{proof}
Non-emptiness follows from the assumption that $|\zeta_l-\eta_l|\le C$ for all $l\in \{1,\ldots,d\}$. For any $l\in \{1,\ldots,d\}$ we have to distinguish the following cases:
    \begin{align*}
    &\zeta_l<\eta_l \implies \beta_l=0,\,\talpha_l=\eta_l=\max(\zeta_l,\eta_l)=\alpha_l\\
    &\zeta_l>\eta_l \implies \beta_l=1,\, \talpha_l=\zeta_l=\max(\zeta_l,\eta_l)=\alpha_l\\
    &\zeta_l=\eta_l \implies \talpha_l=\zeta_l=\eta_l=\max(\zeta_l,\eta_l)=\alpha_l
    \end{align*}
    This yields that $\talpha=\alpha$.
\end{proof}

\begin{lemma}{\textsc{(Odd-Even Representation)}}\label{OddEvenAsLinearConstraints}
Let $\zeta \in \Z^d$ for $d \in \N$. Let $\beta:= \zeta \Mod 2 \in \{0,1\}^d$ and consider the polytope $\mathcal{P}$ in $(\tbeta,\gamma)$ defined by \eqref{maxpolytopeStart2} and \eqref{maxpolytopeEnd2}
\begin{align}
    &\tbeta= \zeta -2\gamma\label{maxpolytopeStart2}\\
    &\tbeta\in \{0,1\}^{d}, \gamma \in \Z^d.\label{maxpolytopeEnd2}
\end{align}
Then it holds that $\mathcal{P}\neq\emptyset$ and every element $(\tbeta,\gamma)\in \mathcal{P}$ satisfies $\tbeta=\beta.$
\end{lemma}
\begin{proof}
Non-emptiness follows since $\zeta \in \Z^d$ per assumption. For any $l\in \{1,\ldots,d\}$ we have to distinguish the following cases:
    \begin{align*}
    &\zeta_l \Mod 2=0\implies \gamma_l=\frac{\zeta}{2},\, \tbeta_l=0=\beta_l\\
    &\zeta_l \Mod 2=1 \implies \gamma_l=\frac{\zeta-1}{2},\, \tbeta_l=1=\beta_l
    \end{align*}
    Thus $\tbeta=\beta$.
\end{proof}

For each bidder $i \in N$ let $\tvi{}: \X \to \R_+$ be a Fourier-sparse approximation of the bidder's reported value function $\pvi{}$. Then the \textit{Fourier transform-based WDP} was defined as follows:
\begin{definition}{(\textsc{Fourier transform-based WDP})}
\begin{align}\label{appendix_FourierTransformBasedWDP}
&\argmax\limits_{a \in \F}\sum_{i \in N}\tvi{a_i}\tag{FT-WDP}.
\end{align}
\end{definition}
Next, we proof Theorem~\ref{FourierTransformBasedMIPs}.
\setcounter{theorem}{0}
\begin{theorem}{\textsc{(Fourier Transform-based MIPs)}}\label{appendix_FourierTransformBasedMIPs}
Let $\tvi{}:\X\to \R$ be a k-Fourier-sparse approximation as defined in \eqref{appendix_FT3}, \eqref{appendix_FT4}, or \eqref{appendix_WHT}. Then there exists a constant $C>0$ such that the MIP defined by the following objective
\begin{align}\label{GenericMIP}
&\argmax\limits_{a \in \F, \beta_i\in \{0,1\}^{k}}\sum_{i \in N}\langle\htvSic,\alpha_i\rangle,
\end{align}
and for $i\in N$ one set of transform specific constraints \eqref{appendix_FT3Contraints1}--\eqref{appendix_FT3Constraints3}, or \eqref{appendix_FT4Contraints1}--\eqref{appendix_FT4Constraints3}, or \eqref{appendix_WHTConstraints1}--\eqref{appendix_WHTConstraints3}, is equivalent to \eqref{appendix_FourierTransformBasedWDP}.
\begin{align}
\textrm{\textbf{FT3:}}\hspace{0.5cm}\textrm{s.t.}\hspace{0.15cm}&\alpha_i\ge 1_k - W_i(1_m-a_i)\label{appendix_FT3Contraints1}\\
&\alpha_i\le 1_k-W_i(1_m-a_i)+C\beta_i\label{appendix_FT3Contraints2}\\
&0_k\le \alpha_i\le C(1_k-\beta_i)\label{appendix_FT3Constraints3}\\
\textrm{\textbf{FT4:}}\hspace{0.5cm}\textrm{s.t.}\hspace{0.15cm}&\alpha_i\ge 1_k - W_i a_i\label{appendix_FT4Contraints1}\\
&\alpha_i\le 1_k-W_i a_i+C\beta_i\label{appendix_FT4Contraints2}\\
&0_k\le \alpha_i\le C(1_k-\beta_i)\label{appendix_FT4Constraints3}\\
\textrm{\textbf{WHT:}}\hspace{0.5cm}\textrm{s.t.}\hspace{0.15cm}&\alpha_i=-2\beta_i+1_k\label{appendix_WHTConstraints1}\\
& \beta_i=W_i a_i-2\gamma_i\label{appendix_WHTConstraints2}\\
&\gamma_i \in \Z^k\label{appendix_WHTConstraints3}
\end{align}
\end{theorem}

\begin{proof}
We proof the equivalence for each of the FTs and corresponding MIPs separately.
\begin{itemize}
    \item \textbf{FT3}\\
    Let $C>0$ such that  $|\left[1_k - W_i(1_m-a_i)\right]_l|\le C$ for all $l\in \{1,\ldots,k\}, i \in N$, and $a \in \F$. We then define  $$\alpha_i:=\max\left(0_k, 1_k - W_i(1_m-a_i)\right)\,\,\textrm{for all }i\in N,$$ and apply for each $\alpha_i$ Lemma \ref{MaxAsLinearConstraints} to $\zeta:=0_k$ and ${\eta:=1_k - W_i(1_m-a_i)}$, which concludes the proof for FT3.\newline
    \item \textbf{FT4}:\\
    This follows analogously as for the FT3 when defining $\alpha_i:=\max(0_k, 1_k - W_i a_i),\, \zeta:=0_k$ and $\eta:=1_k - W_i a_i$.\newline
    \item \textbf{WHT}:\\
    Since for all $i\in N$ and $a\in\F$ it holds that $W_i a_i \in \Z^k$, an immediate consequence from Lemma \ref{OddEvenAsLinearConstraints} is that $\beta_i = W_i a_i \Mod 2$ and $(-1)^{W_ia_i}=-2\beta_i+1_k$, which concludes the proof for the WHT.
\end{itemize}
\end{proof}

\section{Analyzing the Potential of a FT-based CA}

\subsection{Fourier Sparsity - GSVM and MRVM}\label{subsec:appendix_fourier_sparsity}
In this section, we provide the Fourier sparsity results for GSVM and MRVM. In Figure~\ref{fig:energy_distribution_gsvm}, we present the mean over 30 GSVM instances and bidder types. Figure~\ref{fig:energy_distribution_gsvm} shows that GSVM is low degree $\leq 2$ in all considered FTs and thus can be represented using less or equal to $172=\binom{18}{0}+\binom{18}{1}+\binom{18}{2}$ FCs with any FT.

\begin{figure}[t!]
        \centering
        \includegraphics[width=1\columnwidth,trim= 15 15 10 10 ,clip]{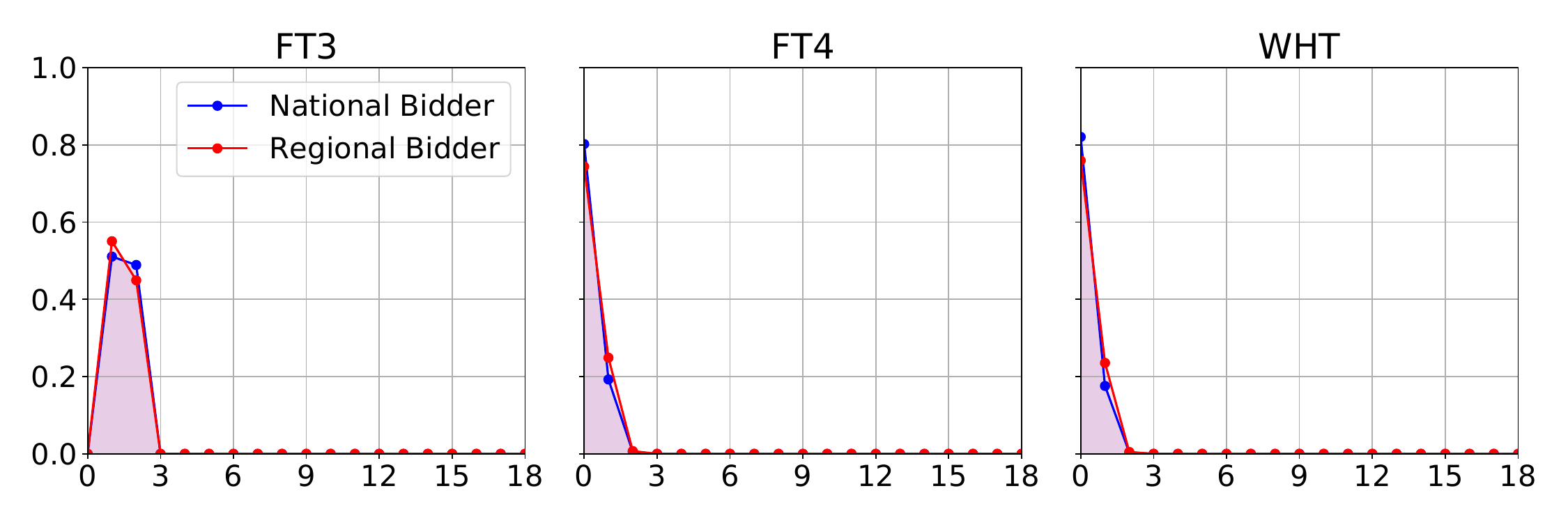}
        \vskip -0.2cm
        \caption{Spectral energy distribution in GSVM for different notions of FTs. For each cardinality (x-axis), we collect the spectral energy (y-axis) of all frequencies of that cardinality and normalize by the total spectral energy.}
        \label{fig:energy_distribution_gsvm}
\end{figure}

To the best of our knowledge it is not possible to compute the exact FCs with respect to any of the considered FTs for all MRVM bidders (see \cite{wendler2020learning}).
In Figure~\ref{fig:energy_distribution_mrvm}, we thus depict the \emph{approximate} WHT-energy-distribution for MRVM. In order to do so, we used the RWHT algorithm by \cited{amrollahi2019efficiently}, which recovers the largest WHT-FCs with high probability. Note that here we do not normalize by the total energy since the total energy is an unknown quantity in the space of $2^{98}$ bundles.

Figure~\ref{fig:energy_distribution_mrvm} shows the \emph{approximate} WHT-energy-distribution for MRVM. We see that while most of the energy is contained in the low degree FCs $\leq2$, there is still some non negligible energy in the FCs of degree 3--7.

\begin{figure}[t!]
        \centering
        \includegraphics[width=1\columnwidth,trim= 15 15 10 10 ,clip]{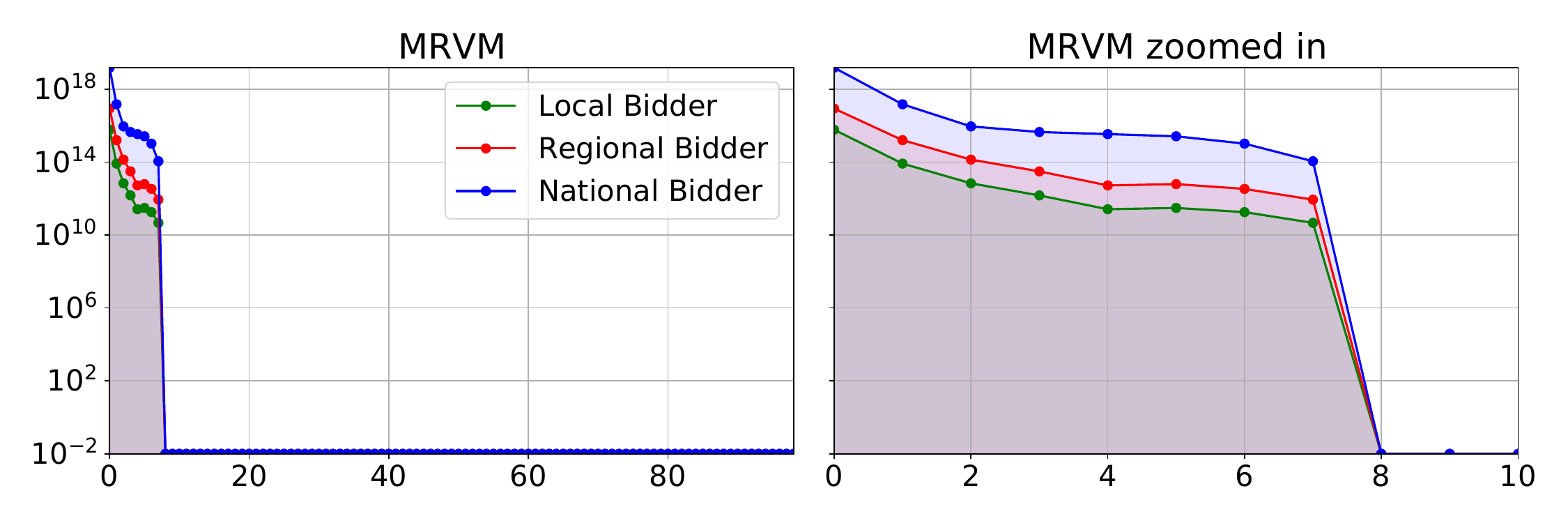}
        \vskip -0.2cm
        \caption{Approximate spectral energy distribution in MRVM for WHT computed with the robust WHT algorithm from \protect\cite{amrollahi2019efficiently}. For each cardinality (x-axis), we collect the spectral energy (y-axis) of all frequencies of that cardinality.}
        \label{fig:energy_distribution_mrvm}
\end{figure}

\subsection{Reconstruction Error of Fourier Transforms}\label{subsec:appendix_reconstruction_error}
Using Procedure~\ref{appendix_Best FCs Full Information Setting}, we determine the best $k$-Fourier-sparse approximation $\tvi{}$ for all considered FTs.

\setcounter{algoprocedure}{0}
\begin{algoprocedure}\textsc{(Best FCs Given Full Access to $\pvi{}$)}\label{appendix_Best FCs Full Information Setting}\ \\
Compute all FCs by using the full FT for each bidders' reported value function $\hvic=F \pvi{}$. Next, 
compute $\tvi{}$ by determining the $k$ best FCs $\htvSic$ in terms of L2 error $\|\pvi{} - \tvi{}\|_2^2$ as follows:
\begin{enumerate}[label=\roman*.]
\item\textbf{WHT:} use the FCs of $\hvic$ with the $k$ largest absolute values.
\item\textbf{FT3 and FT4:} use the FCs with the $k$ largest coefficients $|\hvic(y)|\|F^{-1}_{\cdot,y}\|_2$, where $F^{-1}_{\cdot,y}$ denotes the $y$\textsuperscript{th} column.
\end{enumerate}
\end{algoprocedure}
In the following, we present the technical details for \textit{i.} and \textit{ii.} in Procedure~\ref{appendix_Best FCs Full Information Setting}.

Let $F$ denote the corresponding matrix representation of the FT3, FT4, or the WHT. Furthermore, fix bidder $i \in N$ and for $\Si{1}\subseteq \supp(\hvic),\,|\Si{1}|=k\ll 2^m$ let $$\tvi{}=\sum_{y \in \Si{1}} F^{-1}_{x,y}\hvic(y)=F^{-1}_{\cdot,\Si{1}} \hvSic,$$ 
be a $k$-Fourier-sparse approximation, where we denote by $F^{-1}_{\cdot,\Si{1}}$ the sub matrix of $F^{-1}$ obtained by selecting the columns indexed by the bundles in $\Si{1}$.
\subsubsection{Best Fourier Coefficients WHT}\label{BestFourierCoefficientsFortheWHT}
For the WHT, we consider the following optimization problem of selecting the $k$-best FCs with respect to the quadratic error:
\begin{align}\label{quadraticError}
\Si{1}^{*}\in \argmin\limits_{\Si{1}\subseteq \supp(\hvic),\,|\Si{1}|=k}\|\pvi{} - F^{-1}_{\cdot,\Si{1}} \hvSic\|_2^2.
\end{align}
Then, it follows that $\Si{1}^{*}$ consists out of those bundles $\{y^{(l)}\}_{l=1}^{k}$ with the largest absolute value of the corresponding FCs $\{\hvic(y^{(l)})\}_{l=1}^{k}$. This can be seen as follows.
\begin{align}
    &\|\pvi{} - \underbrace{F^{-1}_{\cdot,\Si{1}} \hvSic}_{\tvi{}}\|_2^2=\|\sum_{y \not\in \Si{1}} F^{-1}_{\cdot,y} \hvic(y) \|_2^2=\\
    =& \sum_{x,y \not\in \Si{1}} \langle F^{-1}_{\cdot,x},F^{-1}_{\cdot,y}\rangle \hvic(x)\hvic(y)=\\
    =&\sum_{y \not\in \Si{1}} \hvic(y)^2,
\end{align}
where in the last equality we used that $F^{-1}$ is an orthogonal matrix and thus its columns fulfill $\langle F^{-1}_{\cdot,x},F^{-1}_{\cdot,y} \rangle=\mathbb{I}_{x=y}$.

For the other non-orthogonal transforms FT3 and FT4, we use a heuristic based on the triangular inequality to select the $k$ ``best'' FCs, which we present next.

\subsubsection{Best Fourier Coefficients FT3 and FT4}\label{BestFourierCoefficientsFT3FT4}
For FT3 and FT4, we use the triangular inequality $$ \|\pvi{} - \underbrace{F^{-1}_{\cdot,\Si{1}} \hvSic}_{\tvi{}}\|_2 \leq \sum_{y \not\in \Si{1}} |\hvic(y)|\|F^{-1}_{\cdot,y}\|_2 $$ to get an upper bound of the quadratic error $ \|\pvi{} - \tvi{}\|_2^2$ and select the FCs with the $k$ largest coefficients $|\hvic(y)|\|F^{-1}_{\cdot,y}\|_2$, where $F^{-1}_{\cdot,y}$ denotes the $y$\textsuperscript{th} column of $F^{-1}$.

\section{A Practical Hybrid ICA Mechanism}

\subsection{Technical Details of \textsc{Hybrid Ica}}\label{subsec:appendix_DetailsOfHybridICA}
In this section, we provide the technical details for the implementation of \textsc{Hybrid Ica}. In particular, we explain in more detail the Fourier Transform-based procedures \ref{Best FCs of Neural Networks}--\ref{FitFSAtoreports}.

\newpage
\subsubsection{Determining the best FCs of NN}
\setcounter{algoprocedure}{2}
\begin{algoprocedure}\textsc{(Best FCs of Neural Networks)}\label{Best FCs of Neural Networks}
\begin{enumerate}[label=\roman*.]
\item If $m\le 29:$ calculate the full FT of the NNs and select the best FCs.
\item If $m>29:$ use sparse FT algorithms, i.e., RWHT for WHT and \textit{SSFT} \cite{wendler2020learning} for FT3 and FT4 to obtain the best FCs of the NNs.
\end{enumerate}
\end{algoprocedure}
Using \textbf{Procedure~\ref{Best FCs of Neural Networks}} we determine the best FCs of the NN estimates of each bidder. We distinguish the two cases:
\begin{itemize}
\item \textbf{Small number of items:} if the number of items is small enough, i.e., $m \leq 29$, we determine the best FCs (and associated locations $\Si{1} \subset \X$) of $\Ni{}$ for all FTs by computing its full Fourier transform $F\Ni{}=\hNi$ using the respective algorithms from \cited{puschel2020discrete}.
\item \textbf{Large number of items:} if the number of items is too large, i.e., $m>29$, we cannot compute the full Fourier transform and thus require sparse FT algorithms to compute the locations of the best FCs. For the WHT, we can do so by using the \textit{robust sparse WHT algorithm (RWHT)} by \cited{amrollahi2019efficiently}. For FT3 and FT4, we use the \textit{sparse set function Fourier transform (SSFT)} algorithm by \cited{wendler2020learning}. Both algorithms compute Fourier-sparse approximations by only using queries from the corresponding set function (= a bidder's value function), and thus do not require a representation of the exponentially large full set function.

\item \textbf{Best WHT FCs:} for the WHT the best FCs are the ones with the largest absolute values (see Section~\ref{BestFourierCoefficientsFortheWHT}).
We select the locations of the $\psuperset$ best NN FCs, i.e., $|\Si{1}| = \psuperset$. We do so because we are going to use the \textit{compressive sensing method} by \cited{stobbe2012learning}, which only requires a superset of the support $\supp (\hvic)$, in \textbf{Procedure \ref{FitFSAtoreports}} to fit a Fourier-sparse approximations $\tvi{}$ to the bidders' reports $R_i$. In our experiments, we set $\psuperset = 2,000$ for all domains. Ideally, we would like to choose $\psuperset$ as large as possible to ensure that the best NN FCs actually overlap with the best FCs of $\pvi{}$, however, there is a trade-off between the number of samples (and also running time) required by the \textit{compressive sensing method} and the size of the support superset. We did not optimize this hyperparameter.

\item \textbf{Best FT3 and FT4 FCs:} for the FT3 and FT4 we use a heuristic based on the triangular inequality to select the best FCs and their locations (see Section~\ref{BestFourierCoefficientsFT3FT4}). As neither FT3 nor FT4 are orthogonal or satisfy the restricted isometry property required by the \textit{compressive sensing method}, we only take the best $\pfr$ locations of NN FCs, i.e., $|\Si{1}| = \pfr$.
\end{itemize}
Notice that the goal of \textbf{Procedure \ref{Best FCs of Neural Networks}} is solely to determine a set of bundles $\Si{1} \subseteq \X$ that is likely to contain many of the dominant FCs of $\pvi{}$. This set $\Si{1}$ is then used to determine reconstructing queries $\Si{2}$.
\newpage
\subsubsection{Determining reconstruction queries}
\begin{algoprocedure}\textsc{(Fourier Reconstruction Queries)}\label{Fourier Reconstruction Queries}
\begin{enumerate}[label=\roman*.]
\item\textbf{FT3 and FT4:} use the sampling theorems by \cited{puschel2020discrete} to determine queries for the bidders, i.e., bundles $\Si{2}\subseteq\X$, that enable a reconstruction of $\tvi{}$.
\item\textbf{WHT:} use the same queries as for FT4.
\end{enumerate}
\end{algoprocedure}
We make use of the sampling theorems for FT3 and FT4 presented by \cited{puschel2020discrete} to obtain reconstruction queries, i.e., queries that help obtaining a small reconstruction error $\|\pvi{} - \tvi{}\|_2$. Our rationale here is that the queries given by the sampling theorem $\Si{2}$ would lead to $\|\pvi{} - \tvi{}\|_2 = 0$ for $\supp(\hvic) = \Si{1}$. The sampling theorem by \cited{puschel2020discrete} selects rows $\Si{2}$ such that $F^{-1}_{\Si{2}, \Si{1}}$ is of full rank, because then (iff $\supp(\hvic) = \Si{1}$) the Fourier coefficients $\htvSic$ are the solution of the linear system of equations ${\hat{v}_{i \mid \Si{2}} = F^{-1}_{\Si{2}, \Si{1}} \htvSic}$. In particular, the theorem yields 
\begin{itemize}
\item \textbf{FT3:} $\Si{2}=\Si{1}$,
\item \textbf{FT4:} $\Si{2} = \set{1_m - y: y \in \Si{1}}$.
\end{itemize}%

\noindent For the WHT there is no sampling theorem and we use 
\begin{align*}
\Si{2} = \set{1_m - y: |\hNi(y)| \text{ is in the }\pfr \text{ largest in } \Si{1}}
\end{align*}
as a heuristic to obtain reconstruction queries. Choosing $\Si{2}$ in that way empirically often leads to a full-rank submatrix $F^{-1}_{\Si{2},\mathcal{Z}}$ of the inverse WHT obtained by selecting the rows indexed by $\Si{2}$ and the columns indexed by $\mathcal{Z}$, where $\mathcal{Z} \subseteq \Si{1}$ are the locations of the $\pfr$ in absolute value largest FCs in $\Si{1}$, i.e., $\mathcal{Z}:= \set{y: |\hNi(y)| \text{ is in the }\pfr \text{ largest in } \Si{1}}$. If $F^{-1}_{\Si{2}, \mathcal{Z}}$ is full rank and $\supp(\hvic) = \mathcal{Z}$, the Fourier coefficients $\phi_{\pvi{}\mid \mathcal{Z}}$ are the solution of the linear system of equations ${\hat{v}_{i \mid \Si{2}} = F^{-1}_{\Si{2}, \mathcal{Z}} \phi_{\pvi{}\mid \mathcal{Z}}}$.  

\subsubsection{Fitting Fourier-sparse approximations}
\begin{algoprocedure}\textsc{\small(Fit Fourier-sparse $\tvi{}$ to Reports)}\label{FitFSAtoreports}
\begin{enumerate}[label=\roman*.]
\item\textbf{FT3 and FT4:} solve the least squares problem defined by the best FCs and reports $R_i$.
\item\textbf{WHT:} use the \textit{compressive sensing method} by \cited{stobbe2012learning} defined by the best FCs and reports $R_i$.
\end{enumerate}
\end{algoprocedure}
Lastly, we need to fit our Fourier-sparse approximations $\tvi{}$ with $\supp(\phi_{\tvi{}}) = \Si{1}$ to the elicited reports $R_i\coloneqq\left\{\left(x^{(l)},\pvi{x^{(l)}}\right)\right\}$. That is, we determine values $w \in \R^{|\Si{1}|}$ for the FCs of the Fourier-sparse approximation $\tvi{} = F^{-1}_{\cdot,\Si{1}} w$.
\begin{itemize}
\item\textbf{WHT:} as already mentioned, for the WHT this is done using \textit{the compressive sensing method} by \cited{stobbe2012learning}. We calculate the full regularization path for the L1-regularization parameter $\lambda$ using the LARS method from \cited{efron2004least}, and select the $\lambda$ that yields $\psupport$ non zero FCs. In our experiments we set $\psupport=100$ in GSVM and LSVM and $\psupport=500$ in MRVM. We did not optimize this hyperparameter.
\item\textbf{FT3 and FT4:} for FT3 and FT4 we cannot use \textit{the compressive sensing method} and instead solve the following least squares problem $$\min_{w \in \mathbb{R}^{\pfr}} \sum_{(x, \pvi{x}) \in R_i} \left(\pvi{x} - \underbrace{\left(F^{-1}_{\cdot,\Si{1}} w\right)(x)}_{=\tvi{x}}\right)^2.$$
\end{itemize}

In Figure~\ref{fig:HybridICADetails}, we present a flow diagram of the different algorithms we use in \textsc{Hybrid Ica}.

\subsection{Details of NNs Support Discovery Experiments}\label{subsec:appendix_NN_support_discovery}
In this section, we discuss the energy ratio plot (see \Cref{fig:energy_ratio}  in the main paper) for the large domain MRVM with $m=98$ items. Recall, that in MRVM one cannot compute the full FT analytically and one has to use sparse FT algorithms to obtain an approximation of the true $k$-best possible frequencies $\mathcal{S}_i^{*} = \{\overset{*}{y}{}^{(1)}, \dots, \overset{*}{y}{}^{(k)}\}$. For this we use the recently developed RWHT algorithm from \cited{amrollahi2019efficiently}, which computes the largest WHT-FCs with high probability and thus gives us an approximation $\mathcal{S}_i^{\text{rwht}}\approx\mathcal{S}_i^{*}$. As in the main paper, we then calculate for each bidder $i\in N$ the \emph{energy ratio} of the support found by the corresponding NNs $\mathcal{\tilde{S}}_i$ and $\mathcal{S}_i^{\text{rwht}}$.

\subsection{Details of Efficiency Experiments}\label{subsec:appendix_efficiency_experiments}

In Figures \ref{fig:appendix_GSVMHybridICAResults}--\ref{fig:appendix_MRVMHybridICAResults}, we present detailed efficiency results of \textsc{Hybrid Ica} in all three SATS domains, where we used the best found configuration of hyperparameters from \Cref{tab:HybridICABestConfig} in the main paper.

In the upper plot we show in each figure the efficiency of the different phases of \textsc{Hybrid Ica}.\footnote{In MRVM we abbreviate the different query types as follows: RI=random initial queries, 1-55=\textsc{Mlca} iterations, FR=Fourier reconstruction queries, FA=Fourier allocation queries.} In the lower plot, we present a histogram of the final efficiency distribution over 100 (in GSVM and LSVM) and 30 (in MRVM) new CA instances. To enable a head-to-head comparison, we use for the test set of CA instances the seeds 1--100 in GSVM and LSVM, and seeds 51--80 in MRVM. 

All experiments were conducted on machines with Intel Xeon E5 v4 2.20GHz processors with 24 cores and 128GB RAM or with Intel E5 v2 2.80GHz processors with 20 cores and 128GB RAM.


\subsubsection{GSVM (Figure~\ref{fig:appendix_GSVMHybridICAResults})}
After $30$ random queries, \textsc{Hybrid Ica} achieves an (avg.) efficiency of $66\%$. Next, after three \textsc{Mlca} iterations, i.e., $21$ \textit{MLCA queries}, 99 instances have an efficiency $\ge90\%$.\footnote{Recall that \textit{MLCA} asks each bidder $n$ queries per iteration.} Finally, the Fourier-based queries significantly increase the (avg.) efficiency by 1.87 percentage points from 98.1\% to 99.97\%.\!\!
Furthermore, we present in the lower plot a histogram of the final efficiency distribution. We see, that for 94 out of 100 instances \textsc{Hybrid Ica} impressively achieves an economic efficiency of 100\% using in total only 100 value queries per bidder.

\subsubsection{LSVM (Figure~\ref{fig:appendix_LSVMHybridICAResults})}
Starting with $\pr=30$ random initial queries, \textsc{Hybrid Ica} achieves an average efficiency of approximately $62\%$. Next, \textsc{Hybrid Ica} performs 5 \textsc{Mlca} iterations. After these 5 \textsc{Mlca} iterations, \textsc{Hybrid Ica} already found for each of the 100 instances an allocation with an efficiency of at least 80\% with an average efficiency of 97.80\%. Here, in the non-sparse LSVM domain, \emph{the Fourier reconstruction and allocation queries} can increase the efficiency of some outliers arriving at an average efficiency of 98.74\%. In the histogram, we see, that for 66 instances, \textsc{Hybrid Ica} was able to achieve full efficiency. Overall, we observe, that in the non-sparse LSVM the Fourier-based approach is not as effective as in the sparse GSVM, but still leads to results, that statistically match the efficiency of \textsc{Mlca} (see \Cref{tab:HybridICAResultsSummary} in the main paper).

\subsubsection{MRVM (Figure~\ref{fig:appendix_MRVMHybridICAResults})}
Starting with $\pr=30$ random initial queries, \textsc{Hybrid Ica} achieves an average efficiency of approximately $50\%$. Next, \textsc{Hybrid Ica} performs 55 \textsc{Mlca} iterations and asks in total 220 \textit{MLCA allocation queries}. Here, we observe a steep increase in efficiency at the beginning. In later iterations the increase in efficiency gets smaller resulting in an average efficiency of $94.20\%$. In MRVM, the best query split we found uses $\pfr = 0$ \emph{Fourier-based reconstruction queries}, thus the Fourier reconstruction queries (FR) do not change the efficiency distribution in Figure~\ref{fig:appendix_MRVMHybridICAResults}. However, the $\pfa = 250$ \emph{Fourier-based allocation queries} significantly increase the efficiency further by $2.43\%$ resulting in a final average efficiency of $96.63\%$.

\begin{figure*}[ht!]
\centering
\resizebox{0.7\textwidth}{!}{
        \includegraphics[width=\textwidth]{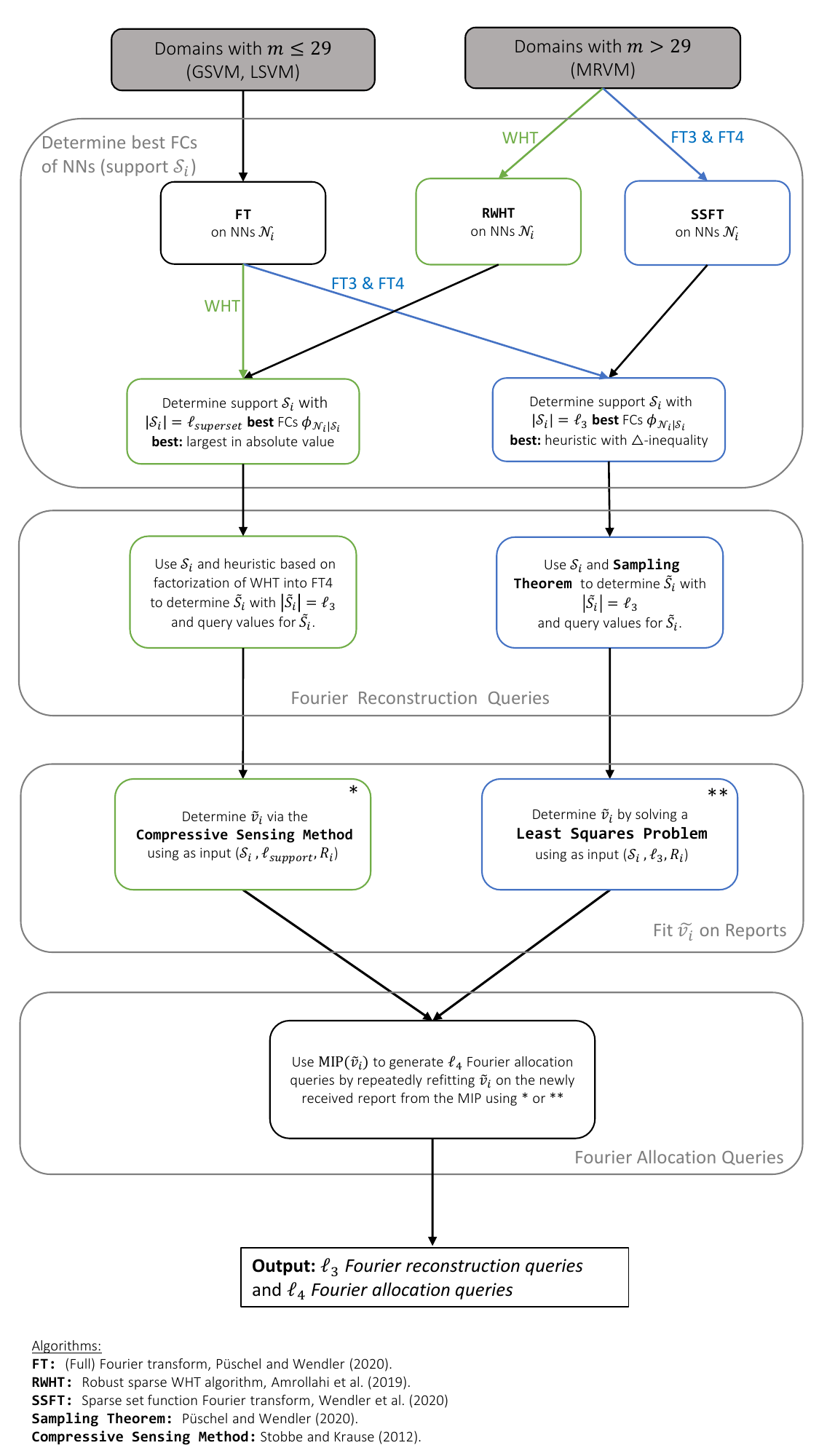}
    }
\caption{Overview of Fourier Transform-based Procedures in \textsc{Hybrid Ica}.}
\label{fig:HybridICADetails}
\end{figure*}

\begin{figure*}[ht]
        \begin{minipage}{.5\textwidth}
        \centering
        \includegraphics[width=1\columnwidth]{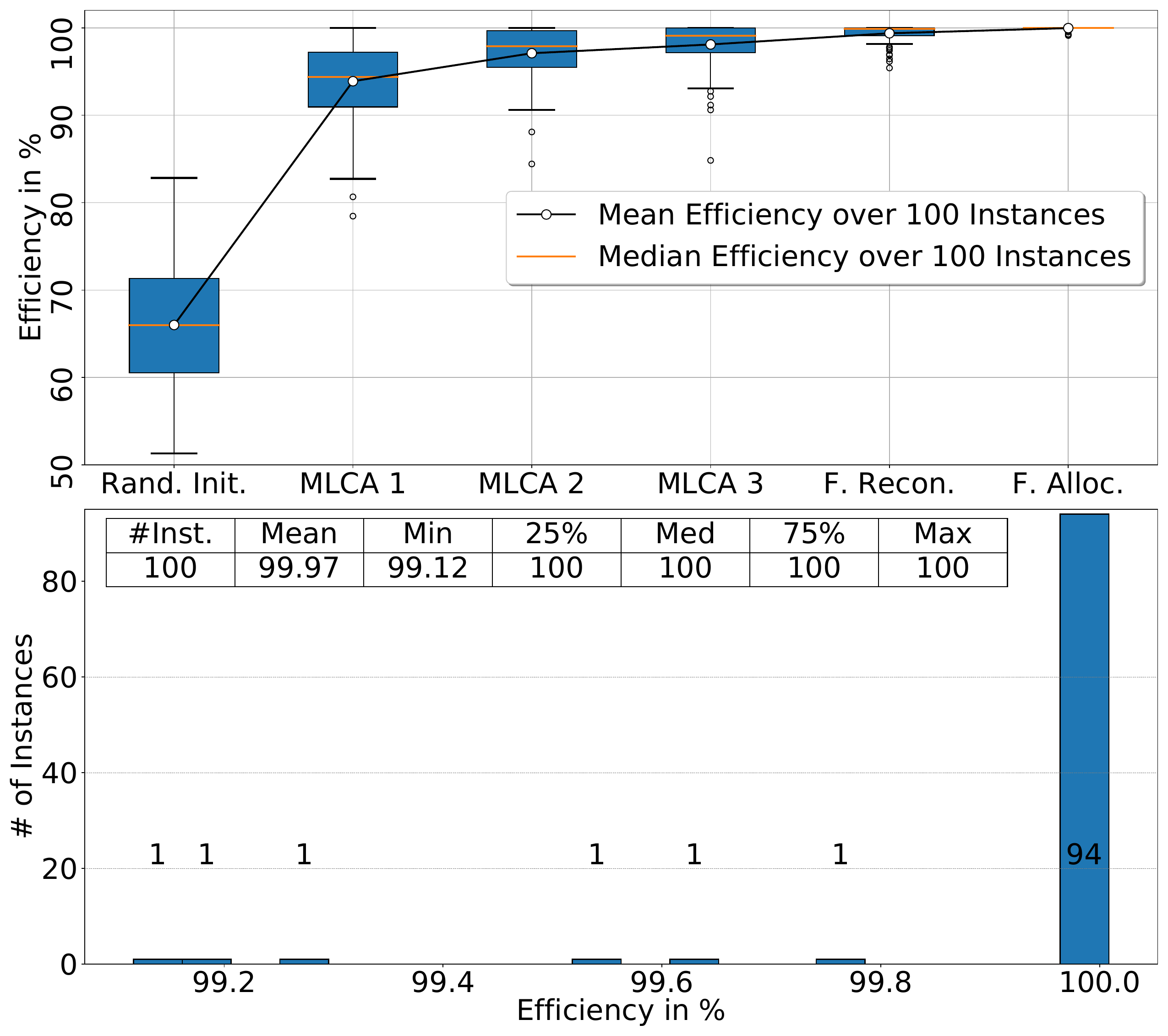}
        \vskip -0.2cm
        \caption{Details of \textsc{Hybrid Ica} in GSVM.}
        \label{fig:appendix_GSVMHybridICAResults}
        \end{minipage}
        \begin{minipage}{.5\textwidth}
        \centering
        \includegraphics[width=1\columnwidth]{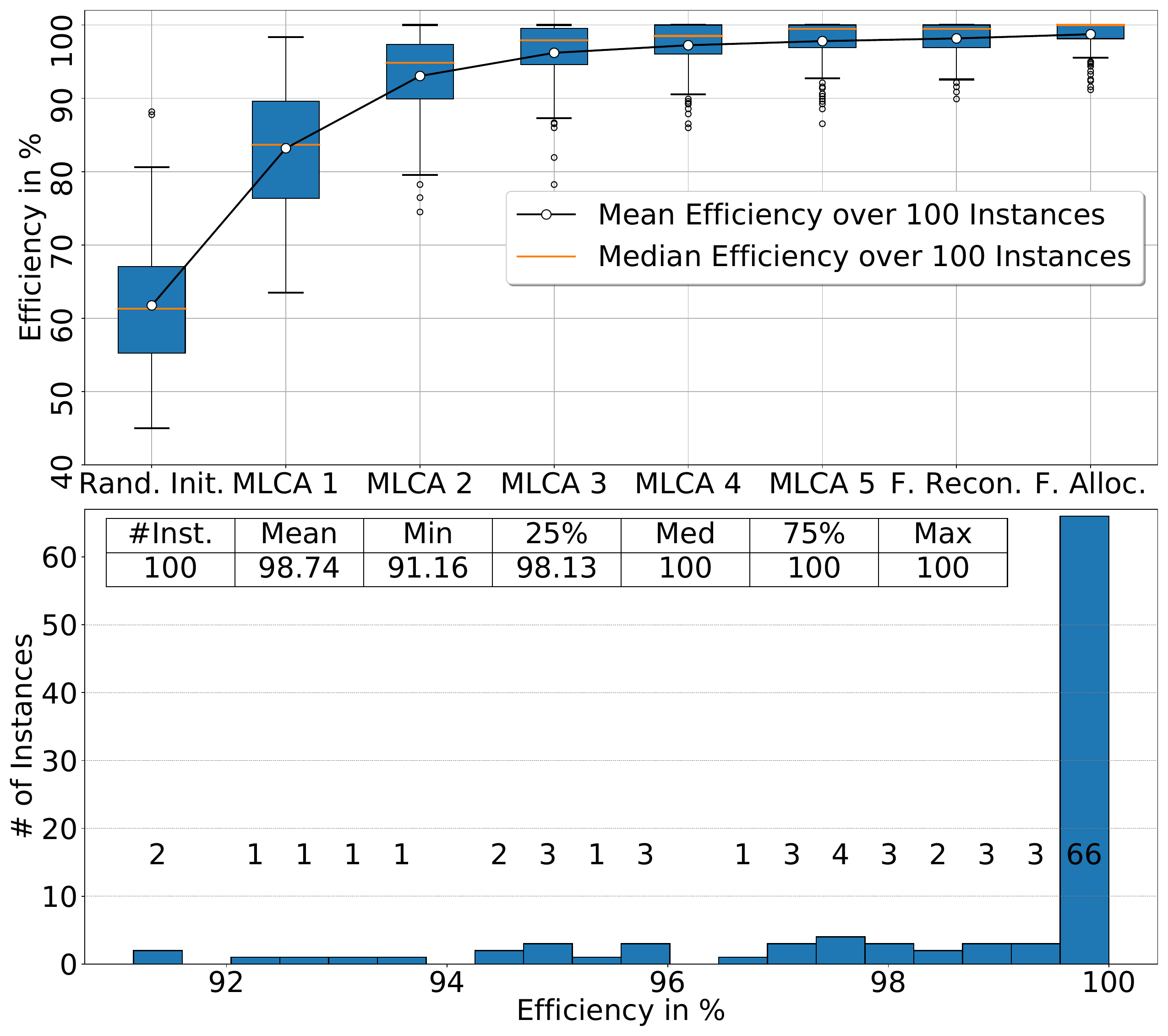}
        \vskip -0.2cm
        \caption{Details of \textsc{Hybrid Ica} in LSVM.}
        \label{fig:appendix_LSVMHybridICAResults}
        \end{minipage}
        \vspace{3cm}\\
        \begin{minipage}{1\textwidth}
        \centering
        \includegraphics[width=1\columnwidth]{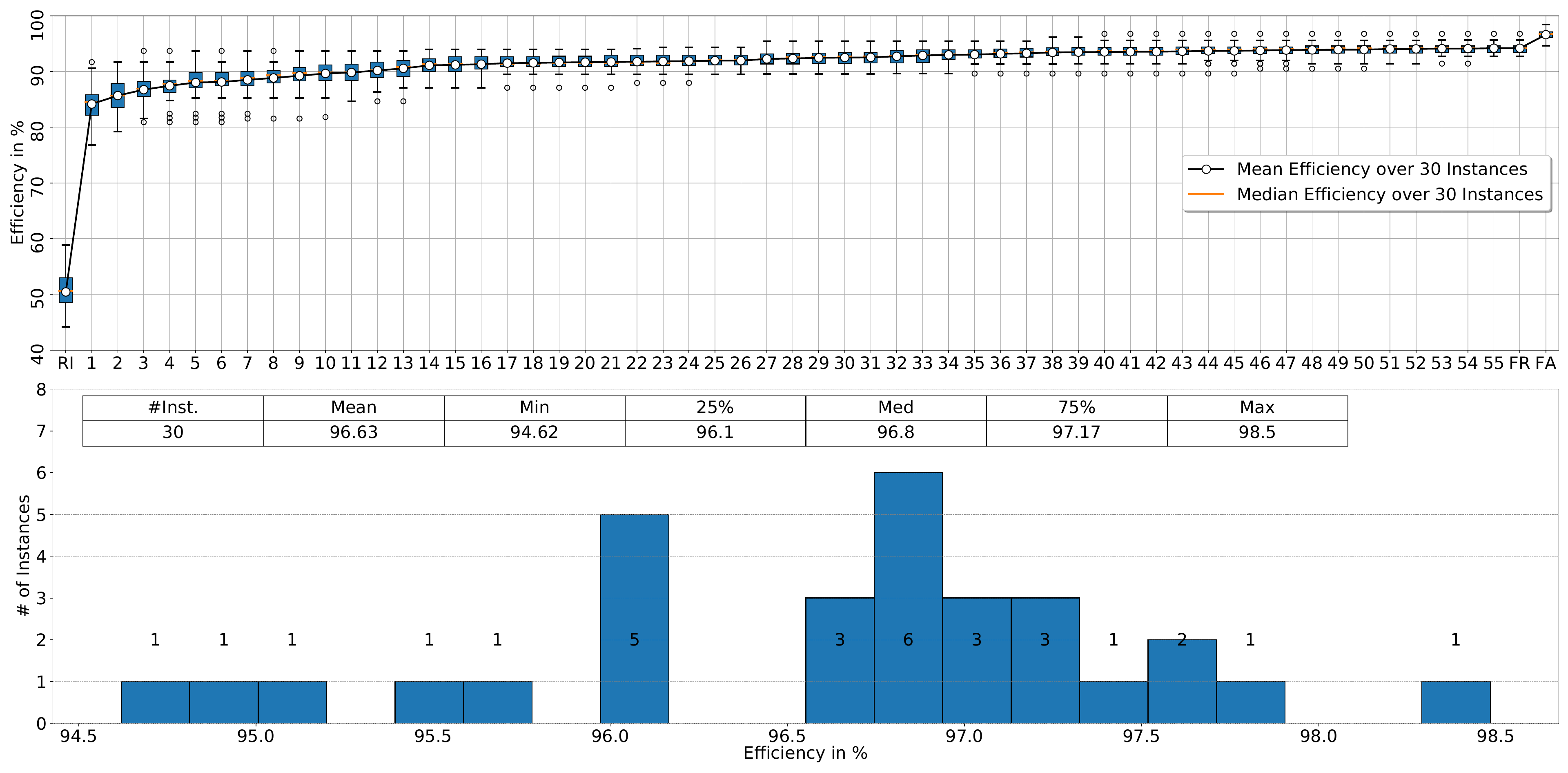}
        \vskip -0.2cm
        \caption{Details of \textsc{Hybrid Ica} in MRVM.}
        \label{fig:appendix_MRVMHybridICAResults}
        \end{minipage}
\end{figure*}
\end{document}